\documentclass[aps,prl,twocolumn,amsmath,amssymb,nofootinbib,showpacs,superscriptaddress]{revtex4-1}
\usepackage[english]{babel}
\usepackage{latexsym}
\usepackage{graphics}
\usepackage{graphicx}
\usepackage{epsfig}
\usepackage{color}
\usepackage{bm}
\usepackage{amsmath}
\usepackage{amssymb}
\usepackage{amsthm}
\usepackage{dcolumn}
\usepackage{bm}
\usepackage{float}
\usepackage{hyperref}
\usepackage{color}
\usepackage{epstopdf}
\usepackage{cleveref}
\usepackage[svgnames]{xcolor}
\usepackage{soul}
\hypersetup{hidelinks,colorlinks=true,allcolors=DarkBlue}

\newcommand{\bM}{\ensuremath{\mathcal{M}}}
\newcommand{\bT}{\ensuremath{\mathcal{T}}}
\newcommand{\bK}{\ensuremath{\mathcal{K}}}
\newcommand{\bH}{\ensuremath{\mathcal{H}}}

\newcommand{\ealm}{$\varepsilon$-almost}
\newcommand{\unit}{universal$_2$}

\newcommand{\eps}{\text{$\varepsilon$}}

\newcommand{\su}{\text{SU${_2}$}}
\newcommand{\xu}{\text{XU${_2}$}}
\newcommand{\eau}{$\varepsilon$\text{-AU}${_2}$}
\newcommand{\easu}{$\varepsilon$-ASU${_2}$}
\newcommand{\eaxu}{$\varepsilon$\text{-AXU}${_2}$}

\newcommand{\otp}{{\rm OTP}}
\newtheorem{defn}{Definition}

\newtheorem{thm}{Theorem}
\newtheorem{lemma}{Lemma}

\newcommand{\KeyPred}{\ensuremath{{\bf KeyPred}}}
\newcommand{\QKD}[1]{\ensuremath{{\bf QKD}_{#1}}}
\newcommand{\Store}{\ensuremath{{\bf Store}}}

\newcommand{\KrecApred}{\ensuremath{\widetilde{K}_{\rm rec}^A}}
\newcommand{\KrecBpred}{\ensuremath{\widetilde{K}_{\rm rec}^B}}

\newcommand{\KrecApredst}{\ensuremath{\overline{K}_{\rm rec}^A}}
\newcommand{\KrecBpredst}{\ensuremath{\overline{K}_{\rm rec}^B}}

\newcommand{\KrecA}{\ensuremath{{K}_{\rm rec}^A}}
\newcommand{\KrecB}{\ensuremath{{K}_{\rm rec}^B}}
\newcommand{\KrecAu}{\ensuremath{\widehat{K}_{\rm rec}^{A}}}
\newcommand{\KrecBu}{\ensuremath{\widehat{K}_{\rm rec}^{B}}}

\newcommand{\KOTPApred}[1]{\ensuremath{\widetilde{K}_{{\rm OTP}(#1)}^A}}
\newcommand{\KOTPBpred}[1]{\ensuremath{\widetilde{K}_{{\rm OTP}(#1)}^B}}
\newcommand{\KOTPApredst}[1]{\ensuremath{\overline{K}_{{\rm OTP}(#1)}^A}}
\newcommand{\KOTPBpredst}[1]{\ensuremath{\overline{K}_{{\rm OTP}(#1)}^B}}

\newcommand{\KOTPA}[1]{\ensuremath{{K}_{{\rm OTP}(#1)}^{A}}}
\newcommand{\KOTPB}[1]{\ensuremath{{K}_{{\rm OTP}(#1)}^B}}
\newcommand{\KOTPAu}[1]{\ensuremath{\widehat{K}_{{\rm OTP}(#1)}^{A}}}
\newcommand{\KOTPBu}[1]{\ensuremath{\widehat{K}_{{\rm OTP}(#1)}^B}}

\newcommand{\Akeyu}[1]{\ensuremath{\widehat{A}_{#1}}}
\newcommand{\Bkeyu}[1]{\ensuremath{\widehat{B}_{#1}}}
\newcommand{\Akey}[1]{\ensuremath{{A}_{#1}}}
\newcommand{\Bkey}[1]{\ensuremath{{B}_{#1}}}

\newcommand{\MA}[1]{\ensuremath{M_#1^\uparrow}}
\newcommand{\MB}[1]{\ensuremath{M_#1^\downarrow}}

\newcommand{\AuthAB}[1]{\ensuremath{{\bf Auth}^{A\rightarrow B}_{#1}}}
\newcommand{\AuthBA}[1]{\ensuremath{{\bf Auth}^{B\rightarrow A}_{#1}}}

\newcommand{\AuthABr}[1]{\ensuremath{{\bf Auth}^{A\rightarrow B, {\rm real}}_{#1}}}
\newcommand{\AuthBAr}[1]{\ensuremath{{\bf Auth}^{B\rightarrow A, {\rm real}}_{#1}}}

\newcommand{\AuthABi}[1]{\ensuremath{{\bf Auth}^{A\rightarrow B, {\rm ideal}}_{#1}}}
\newcommand{\AuthBAi}[1]{\ensuremath{{\bf Auth}^{B\rightarrow A, {\rm ideal}}_{#1}}}

\newcommand{\TQ}[1]{\ensuremath{{\Theta}_{#1}^{\rm QKD}}}
\newcommand{\TA}[1]{\ensuremath{{\Theta}_{#1}^{\rm Auth}}}
\newcommand{\Tamp}[1]{\ensuremath{{\Theta}_{#1}}}
\newcommand{\QE}[1]{\ensuremath{{\mathcal{E}}_{#1}}}

\newcommand{\Suc}[1]{\ensuremath{{\rm Suc}_{#1}}}
\newcommand{\Nsuc}{\ensuremath{N_{\rm suc}}}
\newcommand{\Nmax}{\ensuremath{N_{\rm max}}}

\newcommand{\real}{\ensuremath{^{\rm real}}}
\newcommand{\ideal}{\ensuremath{^{\rm ideal}}}

\newcommand{\ketbra}[2]{|{#1}\rangle_{{#2}}\langle#1|}
\newcommand{\ket}[2]{|{#1}\rangle_{{#2}}}

\newcommand{\Lrec}{\ensuremath{L_{\rm rec}}}
\newcommand{\LOTP}{\ensuremath{L_{\rm OTP}}}
\newcommand{\Lsec}{\ensuremath{L_{\rm sec}}}

\newcommand{\Dist}[2]{\ensuremath{\mathcal{D}(#1,#2)}}
\newcommand{\acc}{\ensuremath{{\rm acc}}}

\begin{document}	
	
\title{Lightweight authentication for quantum key distribution}
\author{E.O.~Kiktenko}
\affiliation{Russian Quantum Center, Skolkovo, Moscow 143025, Russia}
\affiliation{Bauman Moscow State Technical University, Moscow 105005, Russia}
\affiliation{Steklov Mathematical Institute of Russian Academy of Sciences, Moscow 119991, Russia}
\author{A.O.~Malyshev}
\affiliation{Russian Quantum Center, Skolkovo, Moscow 143025, Russia}
\author{M.A.~Gavreev}
\affiliation{Russian Quantum Center, Skolkovo, Moscow 143025, Russia}
\affiliation{Bauman Moscow State Technical University, Moscow 105005, Russia}
\author{A.A.~Bozhedarov}
\affiliation{Russian Quantum Center, Skolkovo, Moscow 143025, Russia}
\affiliation{Skolkovo Institute of Science and Technology, Moscow 121205, Russia}
\author{N.O.~Pozhar}
\affiliation{Russian Quantum Center, Skolkovo, Moscow 143025, Russia}
\author{M.N.~Anufriev}
\affiliation{Russian Quantum Center, Skolkovo, Moscow 143025, Russia}
\author{A.K.~Fedorov}	
\affiliation{Russian Quantum Center, Skolkovo, Moscow 143025, Russia}

\begin{abstract}
Quantum key distribution (QKD) enables unconditionally secure communication between distinct parties using a quantum channel and an authentic public channel.
Reducing the portion of quantum-generated secret keys, that is consumed during the authentication procedure, is of significant importance for improving the performance of QKD systems.  
In the present work, we develop a lightweight authentication protocol for QKD based on a `ping-pong' scheme of authenticity check for QKD.
An important feature of this scheme is that the only one authentication tag is generated and transmitted during each of the QKD post-processing rounds.
For the tag generation purpose, we design an unconditionally secure procedure based on the concept of key recycling. 
The procedure is based on the combination of almost universal$_2$ polynomial hashing, XOR universal$_2$ Toeplitz hashing, and one-time pad (OTP) encryption.
We demonstrate how to minimize both the length of the recycled key and the size of the authentication key, that is required for OTP encryption.
As a result, in real case scenarios, the portion of quantum-generated secret keys that is consumed for the authentication purposes is below 1\%.
Finally, we provide a security analysis of the full quantum key growing process in the framework of universally composable security.
\end{abstract}

\maketitle

\section{Introduction}

QKD is a method for distributing provably secure cryptographic keys in insecure communications networks~\cite{BB84}. 
For this purpose, QKD systems encode information in quantum states of photons and transmit them through optical channels~\cite{Gisin2002,Scarani2009}. 
The security of QKD is then based on laws of quantum physics rather than on computational complexity as is usually the case for public-key cryptography. 
This technology has attracted a significant amount of interest last decades, and industrial QKD systems are now available at retail~\cite{Lo2015,Lo2016,Market}. 

Besides transmitting quantum states in the QKD technology, legitimate users also employ an essential post-processing procedure via an authentic public channel~\cite{Gisin2002,Scarani2009,Lo2015,Lo2016}.
The post-processing procedure consists of key sifting, information reconciliation, privacy amplification, and other supplemental steps.
The authentic classical channel, which prevents malicious modification of the transmitted classical data by an eavesdropper, is essential in order to prevent man-in-the-middle attacks.
The problem of providing classical channel authenticity for QKD has been considered in various aspects in~\cite{Peev2004,Cederlof2008,Peev2011,Abidin2011,Jouguet2011,Mosca2013,Larsson2014,Portmann2014,Larsson2016}.

A conventional approach to solving the authentication problem in QKD systems is to use the Wegman-Carter scheme~\cite{WegmanCarter1981,Stinson1994},
which provides unconditional security.
However, the use of the scheme requires a pair of symmetric secret keys by itself.
Moreover, like in the case of the unconditionally secure encryption with OTP~\cite{Vernam1926}, new secret keys are required for each use of the authenticated channel.
From this perspective, the QKD workflow appears to be a key growing process, since the parties already need to have a short pair of pre-distributed keys before the launching the first QKD round.
For authentication in the second and subsequent rounds, the parts of quantum-generated secret keys from the previous round could be used.
In Ref.~\cite{Quade2009} it was shown that such an approach provides provable composable security of the whole key growing process.
However, an increase in the performance of QKD systems faces a number of challenges, which include the reduction of authentication costs. 
An important task is then to find hash functions that allow minimizing the secret key consumption. 

In our contribution, we present a practical authentication protocol specially designed for minimizing a secret its key consumption.
This goal is achieved by (i) reducing the number of required generations of authentication tags down to one per QKD round and (ii) reducing the size of a secret key consumed by each tag generation. 
The reductions of the key size for tag generations is achieved by the transformation of the delayed authentication scheme previously considered in Ref.~\cite{Peev2009,Maurhart2010,Larsson2016} into a new `ping-pong' scheme, 
where the direction of the tag transmission alters each following QKD round.
In order to minimize the secret key consumption on the tag generation, we consider a combining specially designed almost universal$_2$ polynomial hashing with a standard XOR universal$_2$ Toeplitz hashing followed by the OTP encryption.
This construction allows employing a key recycling approach~\cite{WegmanCarter1981,Portmann2014}, where a permanent `recycling' key is used for polynomial and Toeplitz hashing, while only keys for the OTP encryption require an update.
Then we describe a method for tuning parameters of employed families in order to minimize the length of OTP keys, picked up from each QKD round, while keeping a length of the recycled key, picked up at once from the first QKD round, to be reasonable. 
Finally, we provide a formal security proof of the resulting key growing process in the universally composable security framework that is a standard tool for considering QKD-based key growing process~\cite{BenOr2005}. 
As a result, we obtain a provable secure authentication protocol with very low key consumption, that is why we refer to it as a lightweight authentication protocol. 

Our work is organized as follows. 
In Sec.~\ref{sec:MAC}, we describe the unconditionally secure message authentication code (MAC) scheme with the use of a universal family of hash functions. 
We also review some known approaches for improving their performance. 
In Sec.~\ref{sec:lightweight}, we present our lightweight authentication protocol and show how to derive its parameters with respect to the required security level.
In Sec.~\ref{sec:security} we provide a security analysis of the full key growing process. 
Finally, we summarize the main results and give an outlook in Sec.~\ref{sec:conclusion}.

\section{Unconditionally secure authentication}\label{sec:MAC}

As soon as two legitimate parties, Alice and Bob, have an urge to communicate one each other, they almost inevitably face the problem of assuring that (i) received messages are indeed sent by a claimed sender and
(ii) no adversary can forge messages such that fraud remains undetected. 
This problem is referred to as the authentication problem and it is known in cryptography for ages. 
For example, signatures and seals are ancient yet eligible solutions to \emph{authenticate} handwritten documents and they provide enough intuition about what authentication is. 
Nevertheless, since the topic of this paper lies within the area of digital communication, further we restrict ourselves only to cryptographic authentication schemes~\cite{Schneier1996}. 
From the very beginning, we put ourselves in the framework of unconditional (information-theoretic) security, where we do not rely on any assumption about the computational abilities of an adversary (Eve).

\subsection{Message authentication codes with strongly universal hashing}

A way of providing authentication is to employ the MAC scheme.
The main idea behind it is as follows. 
Suppose that Alice and Bob have a common secret key $k$. 
Then they use the following procedure: 

(i) If Alice wants to authenticate the message $m$, she generates an authentication tag $t$ using the MAC, which is calculated based on $m$ and $k$, and then sends a pair $(m, t)$ to Bob. 

(ii) Given a received pair $(m', t')$, which could be different from $(m,t)$ because of an attack by Eve, Bob generates the corresponding tag $t_{\rm check}$ from $m'$ and $k$, and checks whether the obtained $t_{\rm check}$ is equal to $t'$.

(iii) If so, Bob supposes that the message was indeed sent by Alice and $m'=m$ (and $t'=t$). 

The main point behind this protocol is that, having the intercepted a pair $(m,t)$ but \emph{not} possessing $k$, for Eve it should be practically impossible to generate an alternative valid message-tag pair $(m',t')$ with $m\neq m'$ in order to cheat Bob.
At the same time, it should be impossible for Eve to generate a pair $(m',t')$ for any $m'$ without any authenticated message previously transmitted by Alice.

The described protocol can be realized with the tag generation scheme on the basis of a strongly universal hash family. 
Strongly universal hash-functions have been proposed in Ref.~\cite{WegmanCarter1981} and then formally defined in Ref.~\cite{Stinson1994}. 
A peculiar feature of this method is to employ a whole family of functions with desired properties, but not a single hash function.

Let us consider a set of all possible messages \bM{}, the set of all possible tags \bT{}, and the set of keys \bK{}.
Each key $k\in\bK{}$ defines a function from the family ${h}_{k}:\bM{}\rightarrow\bT{}$.
\begin{defn}[\ealm{} strong \unit{} family of functions]
	A family of functions  
	$\mathcal{H}=\left\{ h_k\colon \mathcal{M}\rightarrow \mathcal{T} \right\}_{k\in\mathcal{K}} $
	is called \ealm{} strongly universal$_2$ (\easu{}) if  for any \emph{distinct} messages $m,m' \in \mathcal{M}$ and \emph{any} tags $t,t' \in \mathcal{T}$  the following two conditions are satisfied:
	\begin{eqnarray}
	&&\underset{k\overset{\$}{\leftarrow}\mathcal{K}}\Pr[h_k(m)=t]=\frac{1}{|\mathcal{T}|},\label{eq:asu1}\\
	&&\underset{k\overset{\$}{\leftarrow}\mathcal{K}}\Pr[h_k(m)=t, h_k(m')=t']\leq \frac{\varepsilon}{|\mathcal{T}|}. \label{eq:asu2}
	\end{eqnarray}
	If $\varepsilon=|\mathcal{T}|^{-1}$, then $\bH$ is called strongly \unit{} (\su{}).
\end{defn}
Here we use $k\overset{\$}{\leftarrow}\mathcal{K}$ for defining a uniformly random generation of an element $k$ from the  set $\mathcal{K}$. 
We note that sometimes condition~\eqref{eq:asu1} is omitted. 

In Ref.~\cite{WegmanCarter1981} it was demonstrated how to construct the family $\mathcal{H}$, and shown how to use it in order to achieve an unconditionally secure authentication.
For this purpose, Alice calculates a tag for message $m$ as $t:=h_k(m)$, where $k$ is a secret key shared by Alice and Bob.
Due to the fact that from the Eve's perspective, the secret key is a uniform random variable, the probability of the impersonation attack, where Eve tries to generate a valid pair $(m',t')$ without any message sent by Alice, is limited by $|\bT|^{-1}$ according to Eq.~\eqref{eq:asu1}.
It is directly follows from Eqs.~\eqref{eq:asu1} and~\eqref{eq:asu2} that for $m'\neq m$ the following expression holds:
\begin{equation}\label{eq:condprob}
	\underset{k\overset{\$}{\leftarrow}\mathcal{K}}\Pr[h_k(m')=t'|h_k(m)=t]\leq \eps.
\end{equation}
Therefore, the probability of the successful realization of the substitution attack, where Eve tries to modify a message (and probably a tag) originally sent by Alice, is limited by $\eps$.

Thus, the use of an \easu{} family with small enough values of $\eps$ allows legitimate parties to achieve unconditionally secure authentication.
The construction issues of appropriate \easu{} families have been studied in various aspects~\cite{WegmanCarter1981,Stinson2002,Boer1993,Bierbrauer1994,Krawczyk1994,Lemire2014}.
Each of these approaches provides a family of hash-functions with its own trade-offs between sizes of $\bM$ and $\bK$, security parameter $\eps$, and efficiency.
Below we consider several common approaches for improving the performance of the authentication scheme based on $\eps$-ASU$_2$ family.

\subsection{Key recycling}\label{sec:recycling}

Here a point of concern is an amount of a key, which is distributed by legitimate parties and required for \easu{}~hashing.
It turns out that in order to achieve unconditional security, Alice and Bob have to use a distinct key $k$ for \emph{each} message, i.e. the key consumption is significant. 
This shortcoming can be avoided. 
In Ref.~\cite{WegmanCarter1981} Wegman and Carter have proposed to choose a key $k$ once.
For each generated tag, we use the XOR operation (bitwise modulo-2 addition) of the resulting tag with a new one-time pad (OTP) key $k_{\rm OTP}$ of $\tau$ bits length, 
where $\tau$ is the tag size (here we assume that $\bT=\{0,1\}^\tau$) and $\tau\ll\log_2{|\mathcal{K}|}$. 
Consequently, in order to authenticate each new message legitimate parties \emph{recycle} the key $k$ and review OTP keys $k_\otp$ only.
Therefore, the key recycling procedure may decrease the demand for secret keys significantly.

Moreover, the key recycling scheme allows using a weaker class of function family, namely, \ealm{} XOR \unit{}. 
\begin{defn}[\ealm{} XOR \unit{} family of hash functions]
	A family of hash functions 
	$\bH=\left\{ h_k\colon \mathcal{M}\rightarrow \{0,1\}^\tau \right\}_{k \in\bK}  $
	is \ealm{} XOR \unit{} (\eaxu{}), if for all distinct $m,m'\in \mathcal{M}$, uniformly random chosen $k\in \bK$ and any $c \in \{0,1\}^\tau$,
	\begin{equation}
	\underset{k\overset{\$}{\leftarrow}\bK}\Pr[h_k(m)\oplus h_k(m') = c] \leq \varepsilon,
	\end{equation}
	where $\oplus$ stands for XOR.
	If $\varepsilon=|\mathcal{T}|^{-1}$, then the family of such hash functions is called XOR universal$_2$ (XU$_2$).
\end{defn}

Consider the following authentication scheme based on an \eaxu{}~family 
$\bH{}= \left\{h_k \colon \bM{} \rightarrow \{0,1\}^\tau \right\}_{k\in\bK}$
with the use of the key recycling scheme.
It is equivalent to the construction of a new family of the following form:
$
\bH^{\rm ext}=\left\{ h_{k,k_\otp}\colon \bM \rightarrow \{0,1\}^\tau \right\}_{k\in\bK,k_\otp\in\{0,1\}^\tau}
$
with $h_{k,k_\otp}(m) := h_k(m) \oplus k_\otp.$
One can see that the resulting family $\bH^{\rm ext}$ is \easu{}, since the requirement given by Eq~\eqref{eq:asu1} follows from the OTP encryption, while the requirement given by Eq.~\eqref{eq:asu2} follows from the fact that
\begin{multline}
	\underset{\substack{k\overset{\$}{\leftarrow}\bK,\\k_{\rm OTP}\overset{\$}{\leftarrow}\{0,1\}^\tau}}\Pr[h_{k,k_\otp}(m)=t,h_{k,k_\otp}(m')=t']\\=\frac{1}{2^\tau}\underset{\substack{k\overset{\$}{\leftarrow}\bK}}{\Pr}[h_{k}(m)\oplus h_{k}(m')=t'\oplus t].
\end{multline}
Thus, the authentication scheme based on the use of an \eaxu{}~family with the key recycling scheme is equivalent to the authentication scheme on the basis of an \easu{} family.
Thus, it allows one to achieve both unconditionally secure authentication and reducing the size of the secret key consumed by each tag generation.

It is important to note that in this scheme the use of the key recycling leads to a disclosure of the used hash function after a number of rounds (the number of rounds can be made arbitrarily large). 
The security proof of the considered scheme in a composable security framework is provided in Ref.~\cite{Portmann2014}.
It is shown that schemes with authenticating $n$ messages based on the key recycling scheme and \eaxu{} family are $n\eps$-secure. 

Following the notion of the composable security framework~\cite{Portmann2014}, it means that an abstract distinguisher with unlimited computational resources, 
has an advantage at most $n\eps$ in distinguishing the considered scheme from an ideal authentication system.
Such a notion of security is particularly important for security analysis of the QKD post-processing procedure.

\subsection{Combining families of hash functions}

A useful technique for the construction of the ASU$_2$ hash family with desired properties is based on the composition of a given family ASU$_2$ and a family with weaker requirements, such as a universal$_2$ family~\cite{WegmanCarter1979}. 
\begin{defn}[\eps-almost \unit{} family of functions]
	A family of functions 
	$
	\bH=\left\{ h_k\colon \bM\rightarrow \bT \right\}_{k\in\bK}
	$
	is called \ealm{} universal$_2$ (\eau{}) if  for any \emph{distinct} messages $m,m' \in \mathcal{M}$ and uniformly chosen $k \in \bK$,
	\begin{equation}\label{eq:au}
	\underset{k\overset{\$}{\leftarrow}\bK}\Pr[h_k(m)=h_k(m')]\leq \eps
	\end{equation}
	If $\varepsilon=|\mathcal{T}|^{-1}$, then $\bH$ is called \unit{} (U$_{2}$).
\end{defn}
A composition of $\eps_1$-AU$_2$ and $\eps_2$-ASU$_2$ families results in  $(\eps_1+\eps_2)$-ASU$_2$ family~\cite{Stinson1994}.
Therefore, having two hash families, one has a useful tool for obtaining a resulting \easu{} family with desired properties.

The similar strategy is applicable in principle  for constructing an \eaxu{} family using the composition $\eps_1$-AU$_2$ and $\eps_2$-AXU$_2$ families.
\begin{thm}\label{thm:1}
	Let	
	$
	\bH^{(1)}=\{h^{(1)}_k:\bM\rightarrow\bT_1\}_{k_1\in\bK_1} 
	$	
	be an $\eps_1$-AU$_2$ family, and 
	$	
	\bH^{(2)}=\{h^{(2)}_k:\bT_1\rightarrow \{0,1\}^\tau\}_{k_2\in\bK_2} 
	$
	be an $\eps_2$-AXU$_2$ family. 
	Then the family
	\begin{equation}
		\bH=\!\{h_{k_1,k_2}=h^{(2)}_{k_2}\circ h^{(1)}_{k_1}:\bM\rightarrow \{0,1\}^\tau\}_{k_1\in\bK_1,k_2\in\bK_2}
	\end{equation}
	is an $(\eps_1+\eps_2)$-AXU$_2$, where $\circ$ denotes a standard function composition. 
\end{thm}

See Appendix~\ref{apx:thm1} for the proof.
Theorem~\ref{thm:1} allows obtaining a $(\eps_1+\eps_2)$-AXU$_2$ hash family as the composition of $\eps_1$-AU$_2$ family with $\eps_2$-AXU$_2$ family and the OTP encryption.

\section{Lightweight authentication protocol for QKD} \label{sec:lightweight}

In this section, we describe an approach for solving the authentication task in the framework of QKD by introducing a lightweight authentication protocol.
The suggested protocol is based on the key recycling scheme and the combination of universal families, which are described above, as well as on a fresh concept of the ping-pong delayed authentication. 

\subsection{Ping-pong delayed authentication}

The workflow of a QKD device can be split into two main stages.
The first ``quantum'' stage is related to the preparing, transmitting and measuring quantum signals (usually, attenuated laser pulses) through an untrusted quantum channel (optical fiber or free space).
As s result of the first stage, the parties, Alice and Bob, obtain a set of records regarding preparing and measuring events.
These records are usually referred to as \emph{raw keys}.
The second stage is the post-processing procedure. 
It is aimed at the extraction (or distillation) of \emph{secret keys} from raw keys, or coming to the conclusion that such an extraction is impossible due to eavesdropper's (Eve's) activities.
The criteria here is that the quantum bit error rate (QBER) value exceeds a certain critical threshold~\cite{Gisin2002}.
In the latter case, the parties just abort the QKD session.

The post-processing stage typically consists of sifting, parameter estimation, information reconciliation, and privacy amplification procedures~\cite{Kiktenko2016}.
During these procedures, the parties communicate with each other via a classical channel, which has to be authenticated in an unconditionally secure way.
It means that any tampering with classical communication by Eve in the classical channel should result in aborting QKD protocol in the same way how it is aborted in the case of its interception in the quantum channel. 
If the QBER value is below the threshold and no tampering in classical communication is detected, Alice and Bob obtain a pair of provably secret keys.
Then the post-processing procedure can be repeated again with a new pair of raw keys.
We refer to such sequences of stage repeating as \emph{rounds}.

The main point of concern is that the unconditionally secure authentication scheme itself requires a symmetric key.
That is why the process of QKD can be considered as a key-growing scheme (or secret-growing scheme):
Alice and Bob have to share some pre-distributed keys in order to provide authentication in the first QKD round and they use a portion of quantum-generated keys for authentication in the following rounds. 
That is why reducing the portion of quantum-generated secret keys, that is consumed during the authentication procedure, is of significant importance for improving the performance of QKD systems.  

\begin{figure}
	\centering
	\includegraphics[width=0.9\linewidth]{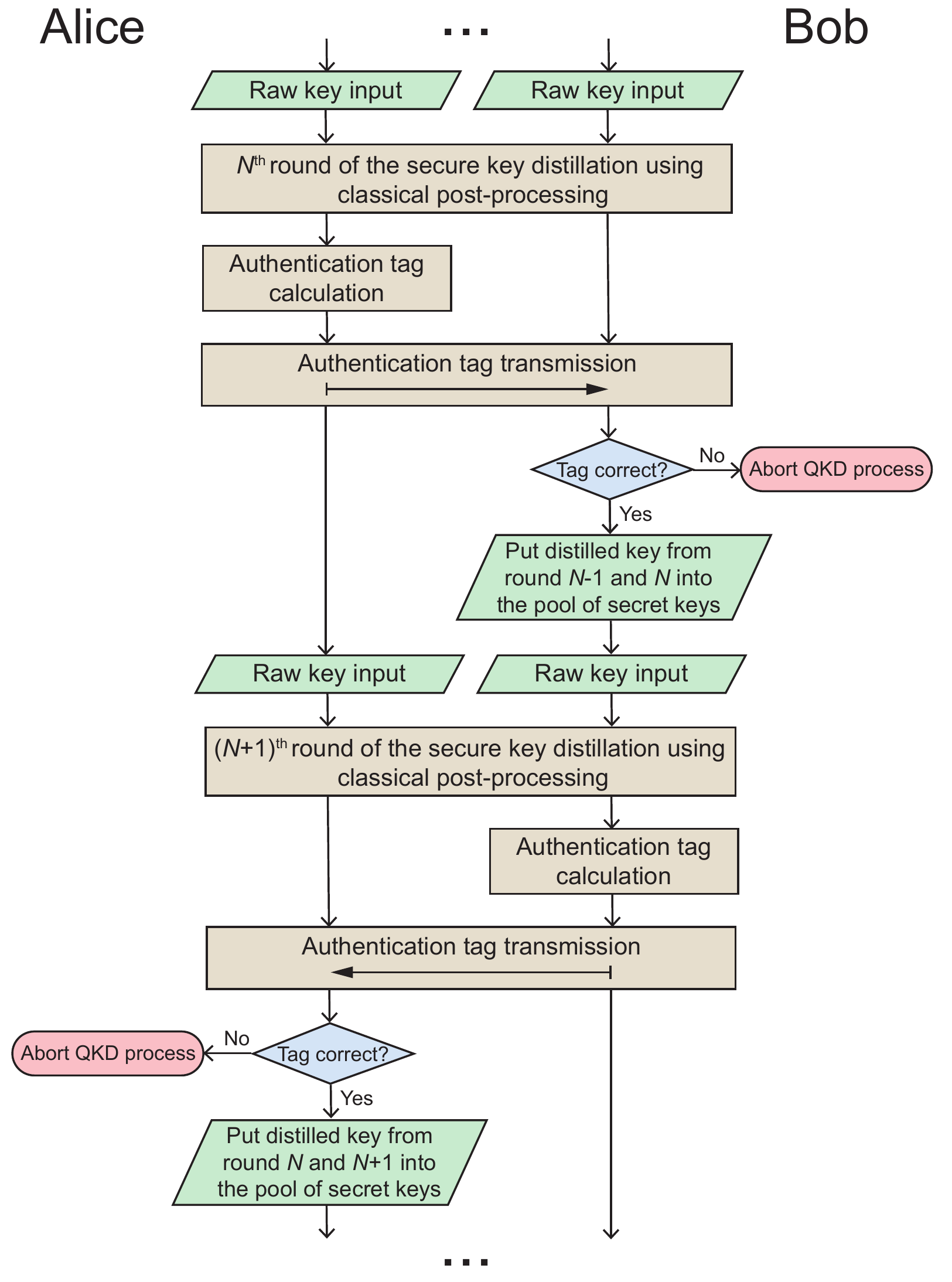}
	\caption{Scheme of the QKD post-processing procedure with delayed authentication. Each of legitimate parties generate and check authentication tags once for the whole traffic during the post-processing procedure.}
	\label{fig:delayed}
\end{figure}

It turns out that it is reasonable to check the authenticity of the classical channel at the very end of the post-processing round, rather than to add an authentication tag to each classical message separately.
The idea is that if Eve has tampered with a classical channel, then the legitimate parties are able to detect this event by authenticating all the traffic in the classical channel once. 
If the authenticity check procedure fails, then secret keys will not be generated.
Meanwhile, in the favorable case, if Eve does not interfere with classical communication, the authentication key consumption will be small.

However, an important point is that, actually, there are two classical channels in which authenticity should be checked: from Alice to Bob and from Bob to Alice.
In order to obtain secure symmetric keys, each of the legitimate parties has to be sure that all the sent and received messages are not modified by Eve.

To address this issue we propose the following solution for authenticity check, that it is applied through pairs of rounds rather than a single round.
The scheme of the designed solution is presented in Fig.~\ref{fig:delayed}.
During each of the post-processing rounds the unconditionally secure authentication tag is transmitted from one party to another in a `ping-pong' manner: 
If in $N^{\rm th}$ round a tag is sent from Alice to Bob, then in $(N+1)^{\rm th}$ round a tag is sent back from Bob to Alice, and so on.
The tag is computed after the final privacy amplification step and is based on a string composed of all the income and outcome classical messages transmitted within the current post-processing round.
After receiving the tag, the legitimate party checks its validity by computing a verification tag from his (her) own versions of sent and received classical messages.
If the check passed, then the party becomes sure that (up to the fixed error probability) the classical communication was not modified by Eve, so the obtained keys are provable secure.
Otherwise, the party terminates the whole QKD process, tries to reach the partner on another channel, and compares the states of secret pools.

In the case of a successful check, the party adds the obtained key into the pool of secret keys.
Moreover, the party also adds to the distilled key from the previous round (if the number of the current round is greater than one).
This is because at the current round the party who receives the tag becomes sure that in the previous round, where he (she) was a tag sender, the authenticity check was also passed.
Otherwise, the protocol should be already terminated on the other side and no valid tag should come at the current round.

Thus, the key consumption is defined by computing a single tag of all the classical messages used in the round.
The exception is a final round where additional authenticated acknowledgement message from a party who performed a check is required.
The generation of this message could be considered as a fictitious post-processing round where no secret key is produced.

\subsection{Instantiation of universal families}

Here we consider the problem of minimizing a key consumption for computing a tag with a given error probability threshold.
For this purpose, we employ the key recycling approach based on the use of an AXU$_2$ family with the OTP encryption (see Subsection~\ref{sec:recycling}).
We construct an AXU$_2$ family by combining AU$_2$ and XU$_2$ families.

Let us consider a value $\eps_{\rm auth}$ that describes an error probability during the authentication procedure. 
On the one hand, with a given fixed $\eps_{\rm auth}$ it is preferable to have a length of the final tag to be as small as possible in order to minimize the OTP key consumption.
This is due to the fact that the length of the tag is equal to one of the OTP keys during each round of the authentication procedure.
Therefore,  it is optimal to use an XU$_2$ family with the minimal possible hash tag length $\tau=-\lceil\log_2\eps_{\rm auth}\rceil$ (hereinafter $\lceil\cdot\rceil$ and $\lfloor\cdot\rfloor$ stand for the standard ceil and floor rounding operations).

On the other hand, it is also important to minimize the length of the recycled key as well.
As it follows from the Stinson bound~\cite{Stinson1996} for $\eps$-AXU$_2$ family,
\begin{equation} \label{eq:Stinson}
	|\bK| \geq \frac{|\bM| (|\bT|-1)}{|\bT|\eps(|\bM|-1)+|\bT|-|\bM|},
\end{equation}
the size of a key, which defines an element from \xu{} family ($\eps= |\bT|^{-1}$), is at least as large as a length of input message. 
It is quite expensive for use in the QKD since message sizes can be about several Mbits.
In order to decrease the length of the required key, we employ a preceding \xu{} hashing with using a function from AU$_2$ family for decreasing the length of the input string.
Such a pre-compression comes together with an additional collision probability $\eps_1$ corresponding to the employed $\eps_1$-AU$_2$ family.
However, we can choose $\eps_1\ll2^{-\tau}$ such that the total error probability $\eps_1+2^{-\tau}$ becomes almost equal to the optimal value $2^{-\tau}$.
As a result, we obtain a scheme that has an optimal OTP key consumption with a recycled key length of the order of $\log\log|\bM|$.
Let us then consider the choice of the particular \eau{} and \xu{} families.

For an ASU$_2$ family we choose a modification of the well-known method of polynomial hashing~\cite{Krovetz2000}.
This approach starts from calculating a polynomial over a finite field with coefficients given by the input and the calculation point given by the random key.
Consider a prime number $p$ of the form $p=2^w+\delta_w,$ where $\delta_w$ is the smallest integer s.t. $p$ is prime for given integer $w$.
Let $\mu$ be an upper bound on a length of an authenticated message, that is a maximal total length of all the messages used in a single round of the QKD post-processing.
Consider a family 
$\bH^{(1)}=\{h^{(1)}_k:\{0,1\}^{\leq\mu}\rightarrow\{0,1\}^{w+1}\}_{k\in\{0,1\}^w}$
with
\begin{equation}
	h_k^{(1)}(m) = \textsf{str}\left(\sum_{i=1}^l\textsf{int}(m_i) \textsf{int}(k)^{i-1}~~({\rm mod}~p) \right).
\end{equation}
Here $\{0,1\}^{\leq\mu}$ stands for a set of all bit-strings of length less or equal to $\mu$, $l:=\lceil (\mu+1)/w\rceil$, \textsf{str} and \textsf{int} are standard functions providing a transition between bit-string and integer number representation, 
and $\{m_i\}_{i=1}^{l}$ are $w$-bit chunks of $m$, which are obtained first by concatenation of $m$ with $1$ followed by a block of zeros in order to achieve an extended string of length $lw$ and then by splitting the resulting string into $l$ of $w$-bit pieces.

A precise statement is then as follows:
\begin{thm}\label{thm:2}
	$\bH^{(1)}$ is an $\eps_1$-AU$_2$ family with
	\begin{equation}\label{eq:epsforpoly}
	\eps_1 =\left \lceil \mu/w \right \rceil 2^{-w}.
	\end{equation}
\end{thm}
See Appendix~\ref{apx:thm2} for the proof. 

From the practical point of view, it is also useful to consider a generalized family of the following form:
$
\bH^{(1)}_\lambda=\{h_{k_1,\ldots,k_\lambda}(\cdot)=h^{(1)}_{k_1}(\cdot)\|\ldots\| h^{(1)}_{k_\lambda}(\cdot):
\{0,1\}^{\leq\mu}\rightarrow\{0,1\}^{\lambda(w+1)}\}_{(k_1,\ldots,k_\lambda)\in\{0,1\}^{w\lambda}}.
$
This family is obtained by the concatenation of $\lambda \geq 1$ independent instances of functions from the family $\bH$.
The following statement holds true:
\begin{thm}\label{thm:3}
	$\bH^{(1)}_\lambda$ is an $\widetilde{\eps}_1$-AU$_2$ family with
	\begin{equation}
		\widetilde{\eps}_1 = \eps_1^\lambda =\left \lceil \mu/w \right \rceil^\lambda 2^{-w\lambda}.
	\end{equation}
\end{thm}
See Appendix~\ref{apx:thm3} for the proof. 

We remind here that the idea behind this generalization is that one can decrease the collision probability down to the desired level without an increase in the employed ring modulus defined by the value of $w$.
Thus, it is possible to set $w=31$ or $w=63$ in order to perform all the calculations with 32 and 64-bit integers, correspondingly.

As the \xu{} family for our protocol, we use Toeplitz hashing~\cite{Krawczyk1994} given by multiplying a message with the Toeplitz matrix in the form:
\begin{equation}
	T_k := \begin{bmatrix}
	k_\beta & k_{\beta+1} & k_{\beta+2} & \ldots & k_{\beta+\alpha-1} \\
	k_{\beta-1} & k_\beta & k_{\beta+1} & \ldots & k_{\beta+\alpha-2} \\
	\vdots & \vdots & \vdots & \ddots & \vdots \\
	k_1 & k_2 & k_3 & \ldots & k_\alpha \\
\end{bmatrix}.
\end{equation}
It is defined by a binary string of the following form: $k=(k_1,\ldots,k_{\alpha+\beta-1})$.
Consider a family 
$
	\bH^{(2)}=\{h_k^{(2)}: \{0,1\}^\alpha\rightarrow\{0,1\}^\beta\}_{k\in \{0,1\}^{\alpha+\beta-1}}, 
$
with
\begin{equation}
	h_k^{(2)}(x)=T_k\cdot x~~({\rm mod}~2),
\end{equation}
where $x$ is treated as the column vector and $\cdot$ stands for the dot-product.
A precise statement is then as follows:
\begin{thm} \label{thm:4}
	$\bH^{(2)}$ is an XU$_2$ family.
\end{thm}
See Appendix~\ref{apx:thm4} for the proof. 

In order to combine Toeplitz hashing with polynomial hashing, we set $\alpha=\lambda(w+1)$ and $\beta=\tau$, where $\tau$ is the length of the final authentication tag.
Consequently, we obtain the authentication scheme, which is illustrated in Fig.~\ref{fig:scheme}.

\begin{figure}
	\centering
	\includegraphics[width=0.7\linewidth]{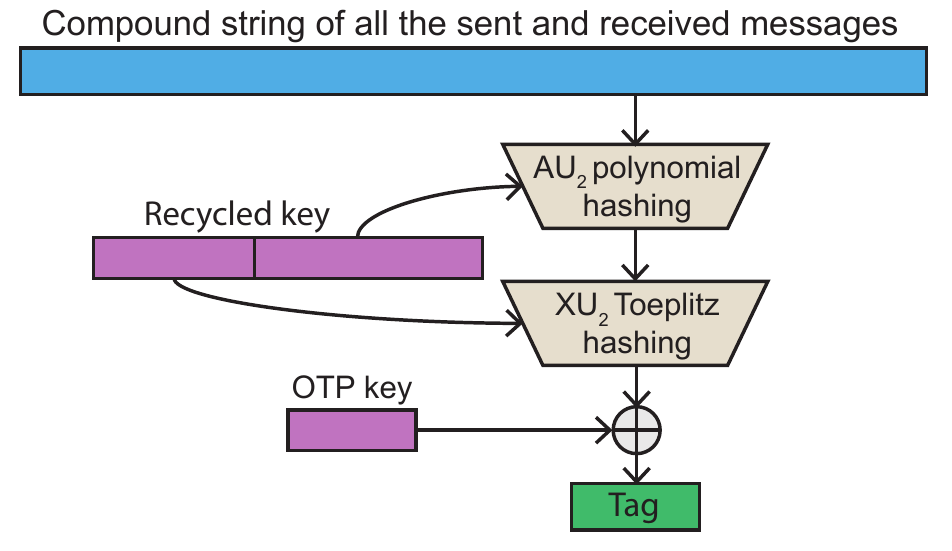}
	\caption{Scheme for the authentication tag calculation in the proposed lightweight authentication protocol for QKD.}
	\label{fig:scheme}
\end{figure}

The length of the recycled key combined from keys for defining $\bH^{(1)}_\lambda$ and $\bH^{(2)}$ elements is given by the following expression: 
\begin{equation}
	L_{\rm rec}=2\lambda w+\lambda+\tau-1,
\end{equation}
while the length of OTP key is given by $L_{\rm OTP}=\tau$.
The final security parameter is as follows:
\begin{equation}
	\eps=2^{-\tau}+\left \lceil \mu/w\right \rceil^{\lambda} 2^{-\lambda w}.
\end{equation}
It grows polynomially with the total size of messages in the classical channel during the post-processing round $\mu$.

Finally, we consider an optimal way of choosing parameters for the proposed authentication scheme.
We start with a given upper bound on tolerable authentication error probability $\eps_{\rm auth}$ and the maximal total length of messages $\mu$.
We also restrict ourselves to consideration only two cases of $w=31$ and $w=63$.

\begin{table}[]\centering
	\begin{tabular}{c|c|c|c|c|}
		\cline{2-5}
	& \multicolumn{2}{c|}{$w=31$} & \multicolumn{2}{c|}{$w=63$} \\ \hline
		\multicolumn{1}{|c|}{$\mu$, Mbits} & $L_{\rm rec}$, bits & $L_{\rm OTP}$, bits & $L_{\rm rec}$, bits & $L_{\rm OTP}$, bits \\ \hline
		\multicolumn{1}{|c|}{1}   & 229 & 40 & 166 & 40 \\ 
		\multicolumn{1}{|c|}{4}   & 291 & 40 & 166 & 40 \\ 
		\multicolumn{1}{|c|}{16}  & 291 & 40 & 166 & 40 \\ 
		\multicolumn{1}{|c|}{64}  & 354 & 40 & 293 & 40 \\ 
		\multicolumn{1}{|c|}{256} & 417 & 40 & 293 & 40 \\ \hline
	\end{tabular}
	\caption{Quantum key consumption on recycled key $L_{\rm rec}$ and OTP key $L_\otp$ for different bound on total length of classical messages transmitted during a single QKD round $\mu$. 
	The security parameter is fixed at the level $\eps_{\rm auth}=10^{-12}$.
	The corresponding Stinson bound [see Eq.~\eqref{eq:Stinson}], which determines a lower theoretical bound on $L_{\rm rec}$ for all presented values of $\mu$, is equal to 44 bits.}
	\label{tbl:performance}
\end{table}

In order to have an the smallest possible OTP key consumption, it is practical to set the following value:
\begin{equation}
	\tau := \lfloor -\log_2\eps_{\rm auth} \rfloor+1.
\end{equation}
The remaining part of $\eps_{\rm auth}$ could be used for AU$_2$ family, so we have to find a minimal possible integer $\lambda$ so that the following condition is fulfilled: 
\begin{equation}
	\lceil \mu/w \rceil^{\lambda} 2^{-\lambda w} \leq \eps_{\rm auth}-2^{-\tau}.
\end{equation}

The performance of the suggested scheme for different values of $\mu$ and $w\in\{31,63\}$ and the fixed value of the security parameter $\eps_{\rm auth}=10^{-12}$ is presented in Table~\ref{tbl:performance}.
We see that for the optimal size of the OTP key, the resulting length of the recycled key is quite low though it can be more than 10 times larger than the one from the bound~\eqref{eq:Stinson}.
Thus, the recycled key can be easily accumulated during the very first QKD round.

\begin{figure} \centering
	\includegraphics[width=1\linewidth]{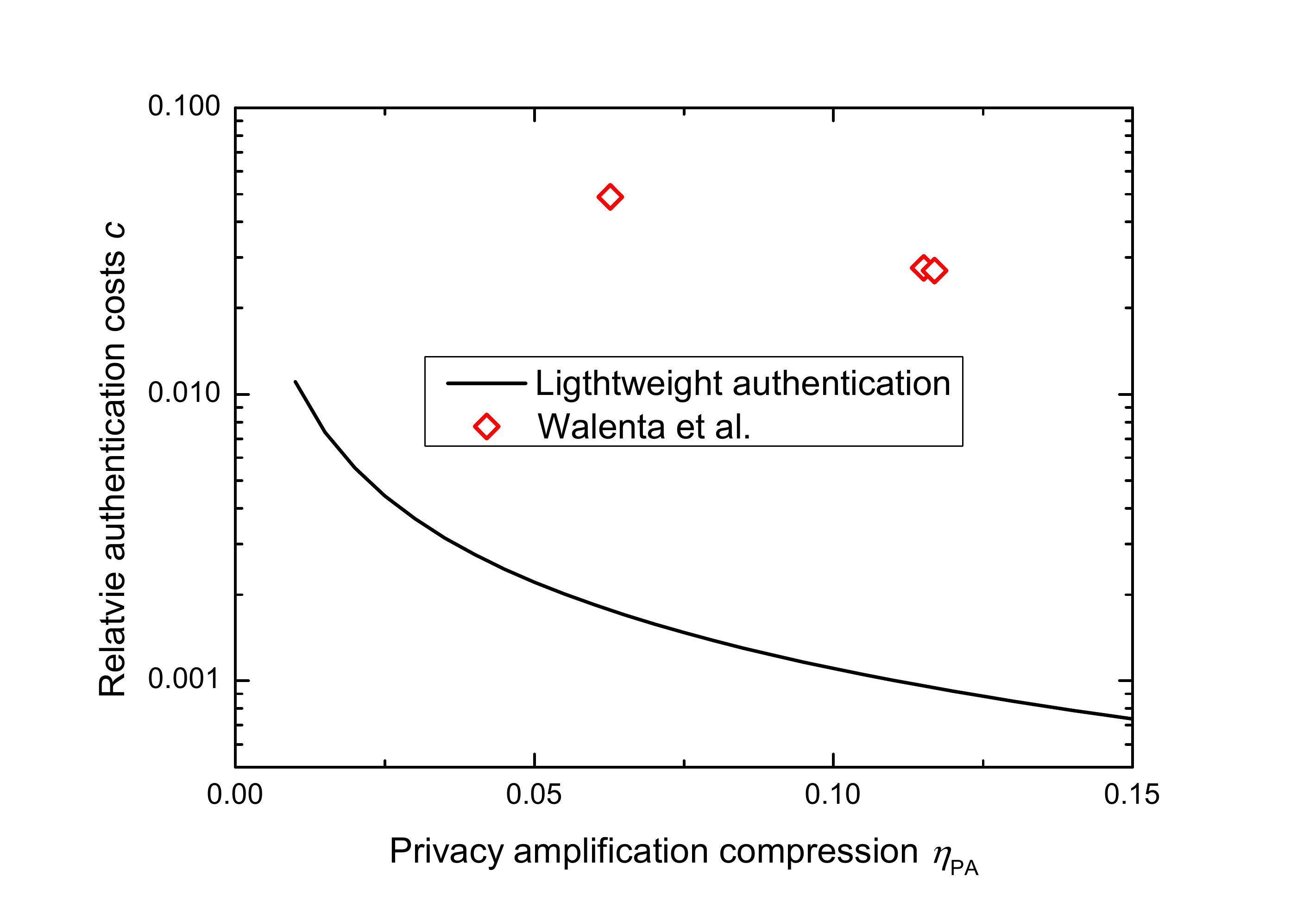}
	\caption{Comparison of the relative authentication costs reconstructed from experimental results presented if Ref.~\cite{Walenta2014} and corresponding values in the case of employing proposed lightweight authentication protocol.
	The security parameter is $\varepsilon_{\rm auth}=10^{-33}$.}
	\label{fig:comparison}
\end{figure}

We also provide a comparison of our approach performance with experimental results on the realization of a fast and versatile QKD system reported in Ref.~\cite{Walenta2014}.
In the considered QKD setup, each round of the post-processing procedure was executed for sifted key blocks of length $L_{\rm sift}=995,328$ bits.
For authentication purposes, the key recycling technique with ASU$_2$ family from Ref.~\cite{Bierbrauer1994} was used.
In the realized authentication scheme a 127-bit tag was generated for every $2^{20}$ bits of classical communication providing the authentication security parameter $\varepsilon_{\rm auth}=10^{-33}$.
As the main figure of merit for the authentication efficiency, we consider a relative authentication cost $c$ defined as a fraction of secret key consumed for authentication in the following rounds.
It has the following form:
\begin{equation}\label{cost}
	c = \frac{L_{\rm OTP}}{L_{\rm sec}} = \frac{\lfloor-\log_2\eps_{\rm auth} \rfloor+1}{L_{\rm sift}\cdot\eta_{\rm pa}}.
\end{equation}
Here $L_{\rm sec}$ is a length of a secret key produced after privacy amplification procedure, and $\eta_{\rm pa}$ is the privacy amplification compression coefficient ($L_{\rm sec}=L_{\rm sift}\cdot\eta_{\rm pa}$).

We present a comparison of the relative authentication costs calculated from the experimental data presented in Table~1 of Ref.~\cite{Walenta2014} with corresponding values for proposed lightweight authentication protocol given by Eq.~(\ref{cost}).
One can see that the lightweight authentication protocol reduces the relative authentication costs by $\approx 28$ times (see Fig.~\ref{fig:comparison}), and for real case scenarios, $c$ is less than 1\%.

We note that the comparison is provided for three points that are extracted from data of Ref.~\cite{Walenta2014} only. 
This is due to the fact that there is no possibility to correctly reconstruct the full dependence of the authentication costs since its size non-trivially depends on the observed level of the QBER.

\section{Security analysis of the key growing process} \label{sec:security}

Here we consider a security analysis of the full key-growing process, which employs the proposed ping-pong authentication protocol.
We first remind that the QKD-based key-growing process is based on the idea that an authentication key for the first QKD round is obtained from some pre-distribution scheme, 
while authentication in next QKD rounds is realized with the use of quantum-generated keys from previous rounds~\cite{BenOr2005}.
Considering the use of the ping-pong authentication scheme with key recycling, we should take into account two important points. 
The first is that the authentication key actually consists of two keys: recycled and OTP keys.
The second is that due to the ping-pong authenticity check pattern, the parties come to the shared pool of secret keys only in the second round, therefore the pre-distributed keys are used in both the first and the second rounds.
Thus, in our setup, we deal with two recycled keys: the pre-distributed one for the first two rounds and the quantum-generated recycled key for all next rounds.
We also have two pre-distributed OTP keys for the first two rounds, and a number of quantum generated OTP keys for authentication in all subsequent rounds.

Consider the general scheme of the \Nmax-round key growing based on the suggested ping-pong authentication scheme [see Fig.~\ref{fig:key_expansion}(a)].
In order to simplify our consideration, we assume that \Nmax~is even, so starting with authentication tag transmission from Alice to Bob in the first round, an authentication tag is transmitted from Bob to Alice in the $\Nmax^{\rm th}$ round.
A peculiar feature of our scheme is that the OTP key for authentication in the $i^{\rm th}$ round comes from the $(i-2)^{\rm th}$ round for $i>2$.
Then in the case of no tampering and no eavesdropping throughout all \Nmax~rounds, the parties also obtain as the output of the process a recycled key with two additional OTP keys, 
which can be used either for launching another key growing procedure or for any other purpose.

\begin{figure*}
	\begin{center}
		\includegraphics[height=0.4\linewidth]{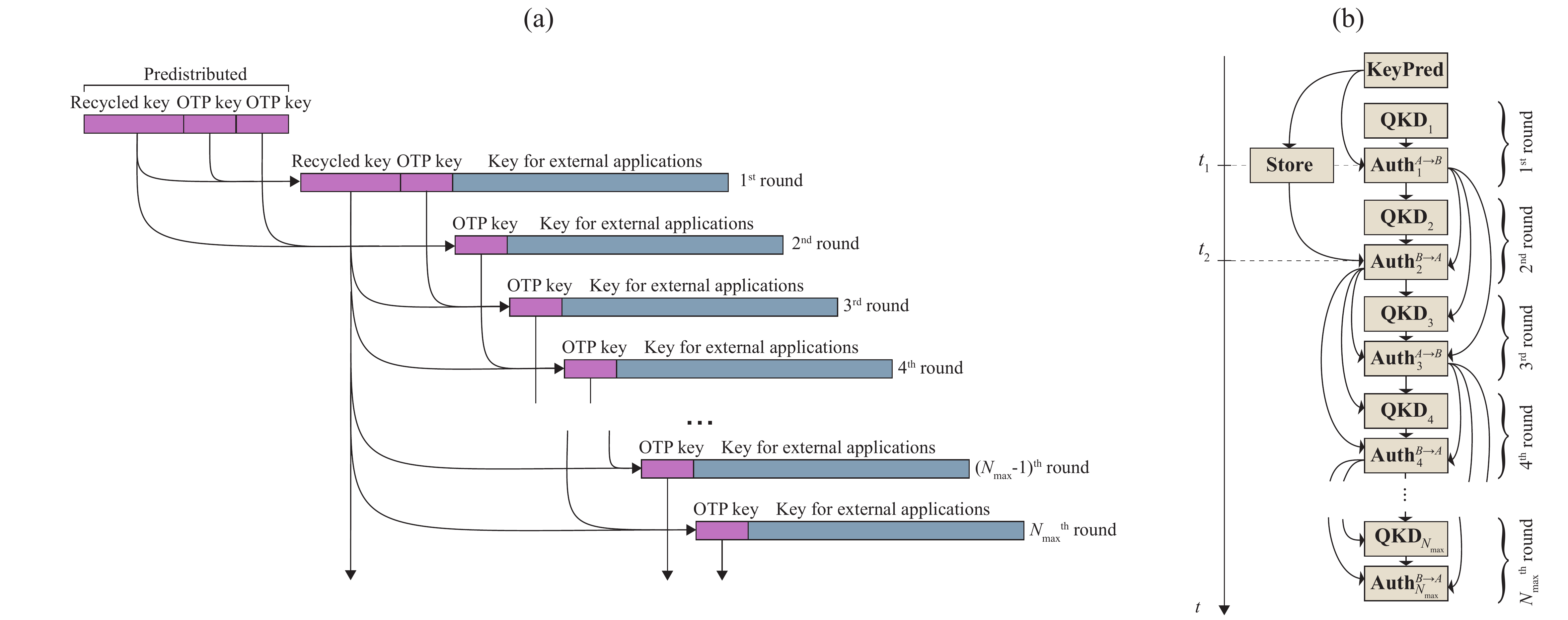}
	\end{center}
	\vskip-4mm
	\caption{
		In (a) the general key growing process with the use of the ping-pong lightweight authentication protocol is presented.
		In (b) the connections of input and outputs of procedures within the key growing process are depicted (see Table~\ref{tab:protocols} for more details).}
	\label{fig:key_expansion}
\end{figure*}

We then describe a general strategy of the security proof of the resulting scheme with skipping technical details.
For complete formal security analysis, we refer the reader to the Appendix~\ref{apx:security_analysis}.

We consider the key growing process as a sequence of procedures.
Each of the procedures is considered as a box with some input and outputs, such that outputs of one box can serve as inputs to the other box [see Fig.~\ref{fig:key_expansion}(b)].
Each of the procedures is considered to have two modifications.
The `\emph{real}' procedure takes place in real-world implementations of the key growing, so it suffers from technical imperfections and has some probability not to cope with its task. 
In contrast, the `\emph{ideal}' procedure copes with its task perfectly, however, it is only a useful abstraction.
We specify a modification of a procedure by a superscript `real' or `ideal' where it is required.

The security proof is based on showing that the keys generated by a sequence of real procedures are indistinguishable from the keys generated by the sequence of ideal procedures up to some small value given by the security parameter.
Since our analysis encounters both classical and quantum systems, we employ a general formalism of quantum states (see Appendix~\ref{apx:braket}) to describe outputs of the procedures.
For quantifying a difference between outputs $\rho_1$ and $\rho_2$, which are semi-positive unit trace linear operators, we use the trace distance of the following form:
\begin{equation}
	\Dist{\rho_1}{\rho_2} := \frac{1}{2}\|\rho_1-\rho_2\|_1,
\end{equation}
where $\|\cdot\|_1$ stands for the trace norm.
It can be interpreted as the maximal possible achievable advantage of distinguishing states $\rho_1$ and $\rho_2$. 
We note that the trace distances reduces to statistical (or total variation) distance if $\rho_1$ and $\rho_2$ describe states of classical systems (random variables).
We describe below ideas of each of the procedures that make up the key growing process.

The $\KeyPred$ procedure performs a pre-distribution of secret key for authentication purposes in the first two rounds.
Remind that it generates an $\Lrec$-bit recycled key and two $\LOTP$-bit OTP keys.
As the output of the procedure, we consider a state of registers storing generated keys on Alice's and Bob's side, and some physical system owned by Eve. 
In the case of $\KeyPred\ideal$ the values of  Alice's and Bob's keys are exactly the same and uniformly distributed over the corresponding key spaces.
Meanwhile, the state of Eve's system is factorized with respect to the state of Alice's and Bob's systems.
It means that Eve knows absolutely nothing about the generated keys.

The output of the real protocol $\KeyPred\real$ is different from the ideal one since some correlations between Eve's system and the pre-distributed keys appear or there can be deviations from a uniform distribution.
Moreover, the output state of the real protocols can depend on time $t$.
This time dependence is intended to capture the fact that Eve can obtain some information about pre-distributed keys with time.
E.g. in the case of the key pre-distribution with classical Diffie--Hellman algorithm~\cite{DH1976}, Eve can wait for a scalable quantum computer and compromise keys by employing the Shor's algorithm~\cite{Shor1997}.
What we assume in our setting is that 
\begin{equation} \label{eq:eps_pred}
	\Dist{\rho_{\rm pred}^{\rm ideal}}{\rho_{\rm pred}^{\rm real}(t_1)}\leq \eps_{\rm pred},
\end{equation}
where $\rho_{\rm pred}^{\rm ideal}$ is the output of the $\KeyPred\ideal$, $\rho_{\rm pred}^{\rm real}(t_1)$ of the $\KeyPred\real$ at the moment $t_1$, 
which is a moment of running the first authenticity check procedure \AuthAB{1}, and $\eps_{\rm pred}$ is some positive real quantity.
Unfortunately, as far as the authors are aware, there is no way to determine $\eps_{\rm pred}$ for traditional (non-quantum) key distribution protocols.
One can only believe that the employed key pre-distribution protocol is good enough to have $\eps_{\rm pred}$ to be less or at least comparable with other $\eps$-s considered further.

The \Store~procedure is introduced to capture the fact that the security of pre-distributed keys generated for the second round may degrade from the time $t_1$ to the moment $t_2$ of the second authentication check procedure \AuthBA{2}.
Its input consists of the pre-distributed recycled keys and the OTP keys for the second round, while the  output consists of the same keys to which is added the state of Eve's system.
In the ideal case , the state of the keys remains the same, and the state of Eve's system does not depend on the values of keys.
In the real case the output state can be different but we assume that 
\begin{equation} \label{eq:eps_store}
	\Dist{\rho^{\rm ideal}_{\rm store}}{\rho^{\rm real}_{\rm store}(t_1,t_2)}\leq \eps_{\rm store},
\end{equation}
with some positive real $\eps_{\rm store}$, where $\rho^{\rm ideal}_{\rm store}$ and $\rho^{\rm real}_{\rm store}(t_1,t_2)$ are outputs of $\Store\ideal$ and $\Store\real$  at the moment $t=t_2$ correspondingly.
We note the output the real procedure depends both on times $t_1$ and $t_2$.

Let us describe a QKD round. 
We split it into two procedures. 
The first procedure is a standard QKD protocol denoted by \QKD{i} which is implemented without any authenticity check.
By this, we mean that all classical communication in a post-processing stage is realized with a public channel and can be potentially tampered.
The output of \QKD{i} includes the generated keys on Alice's and Bob's sides and some physical system owned by Eve.
Due to potential tampering, we call the keys generated in \QKD{i} unverified.
To perform the authenticity check \QKD{i} also outputs resulting logs of classical communication to the second procedure of the QKD round: \AuthAB{i} for odd $i$ and \AuthBA{i} for even $i$.
In these procedures, an authentication tag related to the log is generated, transmitted and verified.
The superscript in the procedure name specifies a direction of the authentication tag transfer.
We refer to the party who generates and transmits the tag as \emph{tag sender}, and to the party who receives and checks the tag as \emph{tag verifier}.

The input of the authenticity check procedure also includes the number of unverified quantum generated keys on the tag verifier's side.
If the verifier does not receive a message with the authentication tag during a pre-specified period of time detects tampering or the received tag is incorrect, then the verifier discards the input unverified keys and terminates the key growing process.
The key growing process is also terminated if the secure key distillation within \QKD{i} failed (e.g. due to the high value of the QBER).
Otherwise, the verifying party turns the unverified keys into verified and continue the process.

In contrast to the key pre-distribution and storing procedure, in the difference between real and ideal instantiations of $\QKD{i}$ and $\AuthAB{i}$ ($\AuthBA{i}$) procedure can be accurately quantified.
Let $\rho^{\rm ideal}_{{\rm QKD}(i)}$ and $\rho^{\rm real}_{{\rm QKD}(i)}$ be the outputs of $\QKD{i}\ideal$ and $\QKD{i}\real$ correspondingly.
Then using existing security proofs for standard QKD protocols (see e.g.~\cite{Tomamichel2017}) we can introduce the following inequality
\begin{equation} \label{eq:eps_QKD}
	\Dist{\rho^{\rm ideal}_{{\rm QKD}(i)}}{\rho^{\rm real}_{{\rm QKD}(i)}} \leq \eps_{\rm QKD},
\end{equation}
where the value of positive real $\eps_{\rm QKD}$ is determined by parameters of the QKD post-processing.
Usually, $\eps_{\rm QKD}$ consists of several terms related to information reconciliation, parameter estimation, and privacy amplification steps and are in the order of $10^{-12}-10^{-9}$ (see e.g.~\cite{Walenta2014,Kiktenko2016}).
A similar situation appears for the outputs of authentication check procedures.
Employing universal composability of authentication with key recycling based on $\eps_{\rm auth}$-ASU$_2$ family proven in Ref.~\cite{Portmann2014}, we argue that
\begin{equation} \label{eq:eps_auth}
	\Dist{\rho^{\rm ideal}_{{\rm auth}(i)}}{\rho^{\rm real}_{{\rm auth}(i)}} \leq \eps_{\rm auth},
\end{equation}
where $\rho^{\rm ideal}_{{\rm auth}(i)}$ and $\rho^{\rm real}_{{\rm auth}(i)}$ are outputs of \AuthABi{i} and \AuthABr{i} (\AuthBAi{i} and \AuthBAr{i}) for odd (even) $i$ correspondingly.
We note that in order to rely on the results obtained in Ref.~\cite{Portmann2014} one have to be sure authentication keys used for tag generation and tag verification are ideal.
We prove a special Lemma in Appendix~\ref{apx:authcheck} which enables us to rely on Eq.~\eqref{eq:eps_auth} in the security proof.

Now we are ready to formulate the main security theorem.
\begin{thm} \label{thm:mainsec}
	Let $\rho^{\Nmax,{\rm ideal}}_{\rm total}$ and $\rho^{\Nmax,{\rm real}}_{\rm total}$ be output states generated in $\Nmax$-round key growing process consisted of all ideal and of all real procedures correspondingly.
	Then the following inequality holds:
	\begin{multline}
		\Dist{\rho^{\Nmax,{\rm ideal}}_{\rm total}}{\rho^{\Nmax,{\rm real}}_{\rm total}}  \\ \leq 
		\eps_{\rm pred} + \eps_{\rm store} + \Nmax(\eps_{\rm auth}+\eps_{\rm QKD}).
	\end{multline}
\end{thm}
See Appendix~\ref{apx:mainsec} for the proof.

The idea of the proof is base on considering the key growing process consisted of all real procedures, consistently replacing real procedures with ideal ones, and calculating distances between appearing intermediate states taking into account~\eqref{eq:eps_pred}-\eqref{eq:eps_auth} (see Fig.~\ref{fig:chain}).
Finally, we note that the security parameter increases linearly with the number of rounds $\Nmax$ that is quite expected for the universal composability framework.

\begin{figure*}
	\centering
	\includegraphics[width=0.8\linewidth]{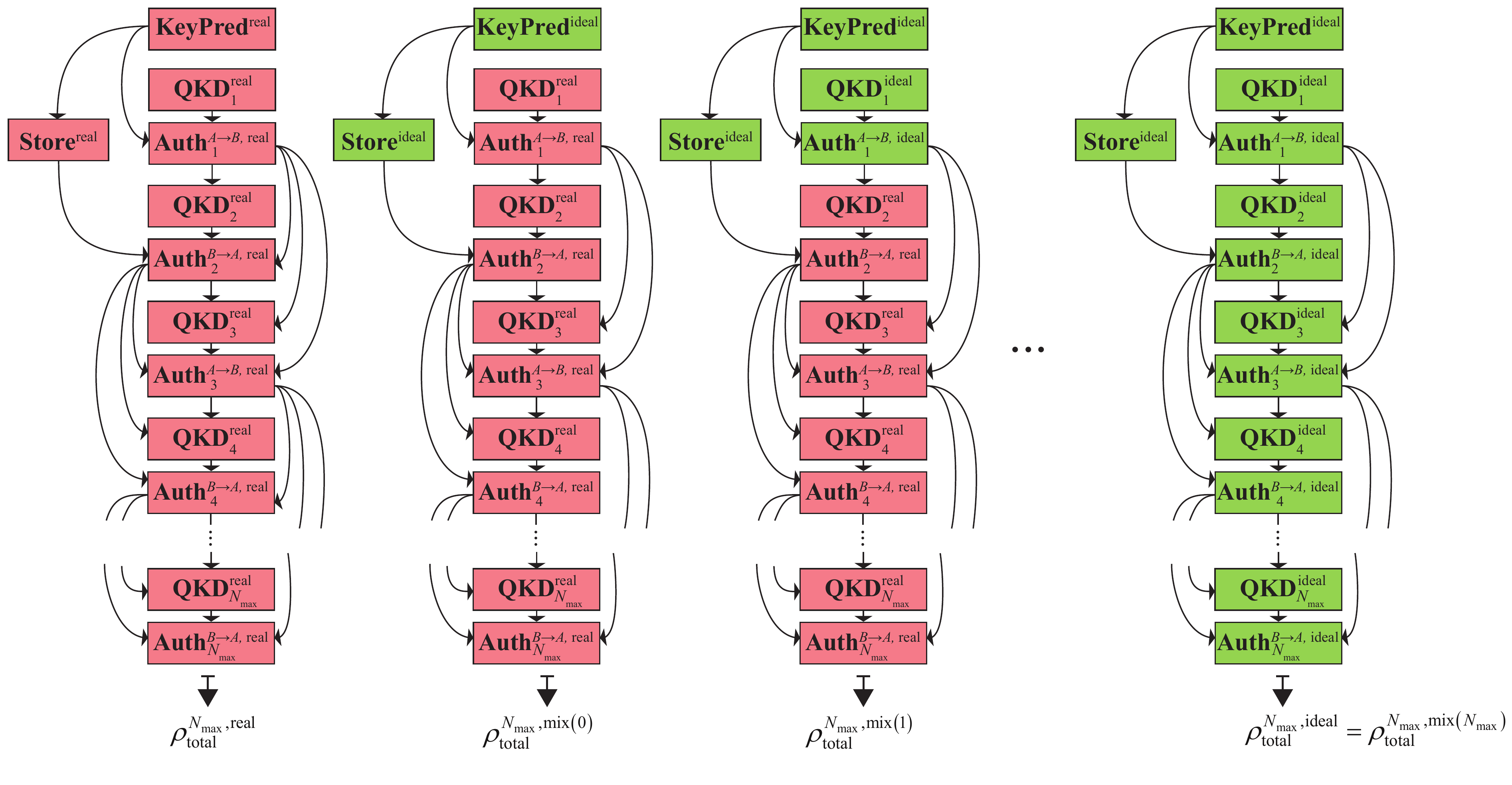}
	\caption{
		Illustration of the basic idea behind the security proof.
		The real procedures are consistently replaced with ideal ones.
		The distance between the states generated in entirely real and entirely ideal processes is obtained as a sum of the distances between states generated in intermediate processes.
	}
	\label{fig:chain}
\end{figure*}

\section{Conclusion and outlook}\label{sec:conclusion}

In this work, we have developed the novel lightweight authentication protocol for QKD systems. 
We have described the practical unconditionally secure authentication scheme for QKD systems, which is based on combining two basic ideas: 
(i) the proposed ping-pong scheme, which allows generating only a single authentication tag during a QKD post-processing round, and (ii) unconditionally secure tag generation based AXU$_2$ family with OTP encryption, 
which allows further reducing key consumption by using key recycling.
We also have demonstrated how to construct suitable AXU$_2$ families using the combination of polynomial AU$_2$ and Toeplitz XU$_2$ families, and
obtained the scheme which minimizes the OTP key consumption for given authentication error bound $\eps_{\rm auth}$.
Finally, we have provided security proof for the key growing process based on our approach.

The proposed scheme is promising for industrial QKD setups due to the possible enhancement of an effective secret key generation rate.Actually, the scheme can be used in other cryptosystems. However, apart from the QKD, the information theoretically-secure authentication is usually overkill. We note that employing the proposed approach also implies the use of one-time-pad encryption to get information-theoretic security all along.

Another interesting area, where the considered approach could be implemented, is QKD networks.
The idea is that in order to generate unconditionally secure secret keys between two parties in a QKD network using a trusted node scheme, all classical communication should be authenticated in an unconditionally secure manner.
For that purpose, a considered implementation of the key recycling approach could be employed.
The source code for a proof-of-principle realization of the suggested authentication protocol is available upon reasonable requests~\cite{Code}.

\section*{Acknowledgment}
The authors would like to thank A.S. Trushechkin and J.-{\AA}. Larsson for fruitful discussions and useful comments.
The work was supported by the Russian Foundation for Basic Research (18-37-20033).

\section{Proofs of the theorems of Sec.~II and III}
\subsection{Proof of the Theorem~\ref{thm:1}}\label{apx:thm1}
\begin{proof}
	Consider two distinct element $m$ and $m'$ from $\bM$.
	We need to prove that for any $c\in\{0,1\}^\tau$, the following relation holds: 
	\begin{equation}~\label{eq:xorun}
	\underset{
		\substack{
			k_1\overset{\$}{\leftarrow}\bK_1\\    
			k_2\overset{\$}{\leftarrow}\bK_2
		}
	}\Pr[h_{k_1,k_2}(m)\oplus h_{k_1,k_2}(m')=c]\leq \eps_1+\eps_2.
	\end{equation}
	
	First, consider $c=(0,\ldots,0)$.
	Then the condition 
	$
	h_{k_1,k_2}(m)\oplus h_{k_1,k_2}(m')=c
	$
	reduces to 
	$
	h_{k_1,k_2}(m)= h_{k_1,k_2}(m'). 
	$
	The LHS of Eq.~\eqref{eq:xorun} can be written as follows:
	\begin{multline}\label{eq:xorun2}
	\underset{
		k_1\overset{\$}{\leftarrow}\bK_1
	}
	\Pr[h^{(1)}_{k_1}(m)= h^{(1)}_{k_1}(m')]
	+
	\underset{
		k_2\overset{\$}{\leftarrow}\bK_2    
	}\Pr[h^{(2)}_{k_2}(h^{(1)}_{k_1}(m))=\\
h^{(2)}_{k_2}(h^{(1)}_{k_1}(m'))
	| h^{(1)}_{k_1}(m)\neq h^{(1)}_{k_1}(m')
	]
	\times\\\underset{
		k_1\overset{\$}{\leftarrow}\bK_1
	}
	\Pr[h^{(1)}_{k_1}(m)\neq h^{(1)}_{k_1}(m')].
	\end{multline}
	The first term in Eq.~\eqref{eq:xorun} is bounded by $\eps_1$ according to the definition of an $\eps_1$-AU$_2$ family, while the second term can be bounded by $\eps_2$ according to the definition of an $\eps_2$-XU$_2$ family.
	Thus, the final upper bound is given by $\eps_1+\eps_2$.
	
	Then consider $c\neq(0,\ldots,0)$.
	One can see that in order to have
	$
	h_{k_1,k_2}(m)\oplus h_{k_1,k_2}(m')=c 
	$
	it is necessary to have 
	$
	h^{(1)}_{k_1}(m)\neq h^{(1)}_{k_1}(m').
	$
	Then we can rewrite LHS of Eq.~\eqref{eq:xorun} in the following form:
	\begin{multline}\label{eq:xorun3}
	\underset{
		k_2\overset{\$}{\leftarrow}\bK_2    
	}\Pr[h^{(2)}_{k_2}(h^{(1)}_{k_1}(m))\oplus h^{(2)}_{k_2}(h^{(1)}_{k_1}(m'))=c |\\
	h^{(1)}_{k_1}(m)\neq h^{(1)}_{k_1}(m')
	]=
	\underset{
		k_1\overset{\$}{\leftarrow}\bK_1
	}
	\Pr[h^{(1)}_{k_1}(m)\neq h^{(1)}_{k_1}(m')].
	\end{multline}
	One can see that the value of Eq.~\eqref{eq:xorun3} is bounded by the value of $\eps_2$ according to the definition of an $\eps_2$-XU$_2$ family, that completes the proof.
\end{proof}

\subsection{Proof of the Theorem~\ref{thm:2}}\label{apx:thm2}
\begin{proof}
	Consider two distinct input bit strings $m$ and $m'$ of length not larger than $\mu$.
	The condition on the collision event is as follows:
	\begin{equation}\label{eq:polynomial}
	\sum_{i=1}^{l}(\textsf{int}(m_i)-\textsf{int}(m_i'))\textsf{int}(k)^{i-1} = 0 ~~({\rm mod}~p),
	\end{equation}
	where $\{m_i\}_{i=1}^l$ and  $\{m_i'\}_{i=1}^l$ are $w$-bit chunks of $m$ and $m'$ correspondingly and $l:=\lceil (\mu+1)/w\rceil$.
	Since $m\neq m'$, there is at least one $m_i\neq m_i'$, so the LHS of Eq.~\eqref{eq:polynomial} can be considered as a non-zero polynomial of the degree that is not larger than $l-1$.
	The collision corresponds to the case, where the  random $k$ taken from $\{0,1\}^w$ turns out to be a root of this polynomial.
	Due to the fact that $(l-1)$-degree polynomial has no more than $l-1$ roots, the collision probability is upper bounded by the following value:
	\begin{equation}
	\frac{l-1}{|\mathcal{K}|}\leq \lceil \mu/w \rceil2^{-w},
	\end{equation}	
	since $\lceil (\mu+1)/w\rceil-1\leq \lceil \mu/w \rceil$.
\end{proof}

\subsection{Proof of the Theorem~\ref{thm:3}} \label{apx:thm3}
\begin{proof}
	Consider two distinct input strings $m$ and $m'$ of length not larger than $\mu$.
	The condition of the identity of 
	\begin{equation}
	h_{k_1,\ldots,k_\lambda}(m)=h_{k_1,\ldots,k_\lambda}(m')
	\end{equation}	 
	is equivalent to a set of identities $h^{(1)}_{k_i}(m)=h^{(1)}_{k_i}(m')$ for all $i=1,\ldots,\lambda$.
	Since all $k_i$ are independent, we obtain
	\begin{multline}
	\underset{(k_1,\ldots,k_\lambda)\overset{\$}{\leftarrow}\{0,1\}^{w\lambda}}{\Pr}
	[h_{k_1,\ldots,k_\lambda}(m)=h_{k_1,\ldots,k_\lambda}(m')]\\=
	\prod_{i=1}^{\lambda}
	\underset{k_i\overset{\$}{\leftarrow}\{0,1\}^{w}}{\Pr}[h_{k_i}(m)=h_{k_i}(m')]\leq\eps_2^\lambda
	\end{multline}
	due to results of the Theorem~\ref{thm:2}.	
\end{proof}

\subsection{Proof of the Theorem~\ref{thm:4}}\label{apx:thm4}
\begin{proof}
	Consider two distinct binary vectors $x$ and $x'$ from $\{0,1\}^\alpha$ and some binary vector $c$ from $\{0,1\}^\beta$.
	Let $\Delta x = x\oplus x'$.
	Denote by $j$ the index of the first nonzero element of $\Delta x$ (we note that $j$ always exists since $x\neq x'$).
	The condition 
	$
	h_k(x)\oplus h_k(x')=c 
	$
	can be rewritten in the following form:
	\begin{equation}\label{eq:Tpl1}
	T_k\cdot \Delta x = c.
	\end{equation}
	One can rewrite Eq.~\eqref{eq:Tpl1} as follows:
	\begin{equation}
	k_{\beta+j-i}+\sum_{\gamma=j+1}^{\alpha}k_{\beta+j-i+\gamma}\Delta x_\gamma=c_i \text{ for } i=1,\ldots,\beta.
	\end{equation}
	Since all the elements of $k$ are independent random bits, it easy to see that for each $i=1,\ldots,\beta$
	\begin{equation}
	\underset{k\overset{\$}{\leftarrow}\bK}{\Pr}[
	k_{\beta+j-i}+\sum_{\gamma=j+1}^{\alpha}k_{\beta+j-i+\gamma}\Delta x_\gamma=c_i ]=\frac{1}{2}
	\end{equation}
	independently from each other.
	Thus, the total probability $\underset{k\overset{\$}{\leftarrow}\bK}{\Pr}[T_k\cdot\Delta x=c]$ is given by $2^{-\beta}$.
\end{proof}

\section{Security proof of the key-growing QKD-based process based on the lightweight authentication scheme} \label{apx:security_analysis}

\newcommand{\bra}[2]{\langle #1 |_{#2}}

\subsection{Basic notations} \label{apx:braket}

In order to to deal with both classical and quantum systems, we use the standard quantum-mechanical formalism.
Consider a (classical) random variable $X$ which takes values from a set $\mathcal{X}$.
Let $\mathbb{H}_{\mathcal{X}}$ be a $|\mathcal{X}|$-dimensional Hilbert space such that every element $x\in\mathcal{X}$ corresponds to a unit basis vector $\ket{x}{X}$ of $\mathbb{H}_{\mathcal{X}}$.
One can think about $\ket{x}{X}$ as column (`ket') vector.
Let $\bra{x}{X}:=\ket{x}{X}^{\dagger}$, where $\dagger$ stands for the Hermitian conjugate (transposition and complex conjugate), be a corresponding row (`bra') vector.

Then the quantum-mechanical description of a random state of the classical variable $X$ takes the following form:
\begin{equation}
	\rho_X:=\sum_{x\in \mathcal{X}}\Pr[X=x]\ketbra{x}{X},
\end{equation}
where $\ketbra{x}{X} \equiv \ket{x}{X}\cdot\bra{x}{X}$.
One can see that $\rho_X$ is the unit trace semi-positive linear operator.

The state of multiple random variables is obtained by the tensor product operation, e.g. a random state of two variables $X$ and $Y$ takes the form $$\rho_{XY}:=\sum_{x\in\mathcal{X},y\in\mathcal{Y}}\Pr[X=x,Y=y]\ketbra{x}{X}\otimes\ketbra{y}{Y},$$
where $\mathcal{Y}$ is a set of possible values of $Y$.

We also introduce some short-hand notations for some common states.
Let 
\begin{equation}
	\sigma^{(l)}_{XY} := \frac{1}{2^l}\sum_{x\in\{0,1\}^l}\ketbra{x}{X}\otimes\ketbra{x}{Y}
\end{equation}
denote a state of two registers $X$ and $Y$ storing a pair of ideal $l$-bit keys.
Let 
\begin{equation}
	\sigma^{(l)}_{X} := \frac{1}{2^l}\sum_{x\in\{0,1\}^l}\ketbra{x}{X}
\end{equation}
be a uniformly-random state of single $l$-bit register $X$.
We also introduce two special symbols $\bot$ and $\acc$, and a state 
\begin{equation}
	\sigma^{\upsilon}_{X_1\ldots X_n}:=\ketbra{\upsilon}{X_1}\otimes\ldots\otimes\ketbra{\upsilon}{X_n}, \quad \upsilon\in\{\bot,\acc\},
\end{equation}
where $\ket{\bot}{X_i}$ and $\ket{\acc}{X_i}$ are two orthogonal normalized states.
Further, we use $\bot$ for indicating errors and $\acc$ for indicating an authentication check passing.

\subsection{Input and outputs of procedures}

We list all inputs and outputs of all procedures employed in the key growing process in Table~\ref{tab:protocols}. 
The descriptions of all variables are summarized in Table~\ref{tab:defs}.
We note that in our consideration the real and ideal modifications have the same inputs and outputs.

\begin{table*}
	\begin{tabular}{|p{0.11\linewidth}|p{0.4\linewidth}|p{0.42\linewidth}|}
		\hline
		Procedure	& 	Input	&	Output \\ \hline	
		\KeyPred	& None	& 
		\KrecApred, \KrecBpred, \KOTPApred{1}, \KOTPBpred{1}, \KOTPApred{2}, \KOTPBpred{2}, 
		$\rho_{E|\KrecApred\KrecBpred\KOTPApred{1}\KOTPBpred{1}\KOTPApred{2}\KOTPBpred{2}}$ \\ \hline
		
		\QKD{1} & None &
		\Akeyu{i}, \Bkeyu{i}, \KOTPAu{i+2}, \KOTPBu{i+2}, \MA{i}, \MB{i}, \KrecAu, \KrecBu, \newline
		$\rho_{E|\Akeyu{1}\Bkeyu{1}\KOTPAu{i+2}\KOTPBu{i+2}\MA{i}\MB{i}\KrecAu\KrecBu}$  \\\hline
		
		\QKD{i}, \newline $i\geq 2$ & $V_{i-2}$, $V_{i-1}$, ($V_0:=\acc$) &
		\Akeyu{i}, \Bkeyu{i}, \KOTPAu{i+2}, \KOTPBu{i+2}, \MA{i}, \MB{i}, \newline
		$\rho_{E|\Akey{i}\Bkeyu{i}\KOTPAu{i+2}\KOTPBu{i+2}\MA{i}\MB{i}}$  \\\hline
		
		\AuthAB{1} & \KrecApred, \KrecBpred, \KOTPApred{1}, \KOTPBpred{1}, \MA{1}, \MB{1},
		\Bkeyu{1}, \KrecBu, \KOTPBu{3} &
		$V_1$, \Bkey{1}, \KrecB, \KOTPB{3}	\\ \hline
		
		\AuthBA{2} & $V_{1}$, \KrecApredst, \KrecBpredst, \KOTPApredst{2}, \KOTPBpredst{2}, \MA{2}, \MB{2},
		\Akeyu{1}, \KrecAu, \KOTPAu{3}, \Akeyu{2}, \KOTPAu{4} &
		$V_2$, \Akey{1}, \KrecA, \KOTPA{3}, \Akey{2}, \KOTPA{4}	\\ \hline
		
		\AuthAB{i}, $i=2j+1$, \newline $j\geq2$ & $V_{i-1}$, \KrecA, \KrecB, \KOTPA{i}, \KOTPB{i}, \MA{i}, \MB{i},
		\Bkeyu{i-1}, \KOTPBu{i+1}, \Bkeyu{i}, \KOTPBu{i+2} &
		$V_i$, \KrecA, \KrecB, \Bkey{i-1}, \KOTPB{i+1}, \Bkey{i}, \KOTPB{i+2}	\\ \hline
		
		\AuthBA{i}, $i=2j, j\geq2$ & $V_{i-1}$, \KrecA, \KrecB, \KOTPA{i}, \KOTPB{i}, \MA{i}, \MB{i},
		\Akeyu{i-1}, \KOTPAu{i+1}, \Akeyu{i}, \KOTPAu{i+2} &
		$V_i$, \KrecA, \KrecB, \Akey{i-1}, \KOTPA{i+1}, \Akey{i}, \KOTPA{i+2}	\\ \hline
		
		\Store & \KrecApred, \KrecBpred, \KOTPApred{2}, \KOTPBpred{2} &	\KrecApredst, \KrecBpredst, \KOTPApredst{2}, \KOTPBpredst{2}, \newline
		$\rho_{E|\KrecApredst\KrecBpredst\KOTPApredst{2}\KOTPBpredst{2}}$ \\ \hline
	\end{tabular}
	\caption{List of all inputs and outputs of procedures used in the key growing process.  
		For descriptions of variables see Table~\ref{tab:defs}.}
	\label{tab:protocols}
\end{table*}

\begin{table*}
	\begin{tabular}{|p{0.16\linewidth}|p{0.8\linewidth}|}
		\hline
		Variable & Description \\ \hline
		\KrecApred (\KrecBpred) & Alice's (Bob's) pre-distributed recycled key for authentication in the first two rounds. \\
		\KrecApredst (\KrecBpredst) & Alice's (Bob's) pre-distributed recycled key for authentication in the second round after the storage period. \\
		\KrecA (\KrecB) & Alice's (Bob's) quantum generated recycled key for authentication in all rounds after the second. \\
		\KOTPApred{i} (\KOTPBpred{i}) & Alice's (Bob's) pre-distributed OTP key for authentication in the $i^{\rm th}$ round. \\
		\KOTPApredst{2} (\KOTPBpredst{2}) & Alice's (Bob's) pre-distributed OTP key for authentication in the $2^{\rm nd}$ round after the storage period. \\
		$V_{i}$ & Verification flag of the authenticity check obtained in $i^{\rm th}$ round by Bob (Alice) for odd (even) round number $i$. $V_{i}=\acc$ if the authenticity check passed, and $V_{i}=\bot$ -- otherwise.\\
		\Akeyu{i} (\Bkeyu{i}) & Alice's (Bob's) unverified quantum generated key in $i^{\rm th}$ round for external applications. \\
		\Akey{i} (\Bkey{i}) & Alice's (Bob's) verified quantum generated key from $i^{\rm th}$ round for external applications. \\
		\MA{i} (\MB{i}) & Compound string of the received and transferred (transferred and received) classical messages during post-processing in $i^{\rm th}$ round on Alice's (Bob's) side for odd $i$ and Bob's (Alice's) side for even $i$. \\
		$\rho_{E|X_1X_2\ldots}$ & Quantum state of Eve's system of arbitrary physical  nature conditioned by values of $X_1, X_2, \ldots $ \\\hline
	\end{tabular}
	\caption{Descriptions of variables used in the key growing process.}
	\label{tab:defs}
\end{table*}

\subsection{Formal description of procedures}
\subsubsection{\KeyPred}
This procedure performs a pre-distribution of secret key for authentication purposes in the first two rounds.
In the ideal case its output is given by
\begin{equation}
	\rho_{\rm pred}^{\rm ideal} := \sigma^{(\Lrec)}_{\KrecApred\KrecBpred}\otimes
	\sigma^{(\LOTP)}_{\KOTPApred{1}\KOTPBpred{1}}\otimes\sigma^{(\LOTP)}_{\KOTPApred{2}\KOTPBpred{2}}\otimes
	\rho_E^{\rm pred},
\end{equation}
where $\rho_E^{\rm pred}$ is some state of quantum system belonging to Eve. 
We note that $\rho_E^{\rm pred}$ comes as a tensor factor, so Eve obtains no information about generated keys possessed by Alice and Bob.

As it was mentioned in the main text, the output of the real protocol $\rho_{\rm pred}^{\rm real}(t)\leftarrow \KeyPred\real$ can be different from the ideal one and depends on time $t$.
However, we assume that this difference at the moment $t=t_1$ is limited by the value of $\eps_{\rm pred}$ in the Eq.~\eqref{eq:eps_pred}.

\subsubsection{\Store}
This procedure just stores pre-distributed keys for authentication check of the second round \KrecApred, \KrecBpred, \KOTPApred{2}, and \KOTPBpred{2} until the corresponding authentication tag is generated and (or) verified.
The output of the $\Store\ideal$ is given by 
\begin{multline}
	\rho^{\rm ideal}_{\rm store}= \ketbra{\KrecApred}{\KrecApredst}\otimes  \ketbra{\KrecBpred}{\KrecBpredst} \otimes \\ 
	\ketbra{\KOTPApred{2}}{\KOTPApredst{2}}\otimes  \\ \ketbra{\KOTPBpred{2}}{\KOTPBpredst{2}}\otimes \rho_{E}^{{\rm store}},
\end{multline}
where $\rho_{E}^{{\rm store}}$ is some state of Eve's system which  does not depend on input keys values.

Like in the case of \KeyPred\real, the output state  $\rho^{\rm real}_{\rm store}(t_1,t_2)\leftarrow\Store\real$ is considered to depend on duration of storing defined by $t_1$ and $t_2$.
We assume that the difference between real and ideal storing procedures at the moment $t=t_2$ is limited by $\eps_{\rm store}$ in the Eq.~\eqref{eq:eps_store}.

\subsubsection{\QKD{i}}
This procedure corresponds to running a standard QKD protocol by Alice and Bob, whenever a tampering of the classical channel was not detected in previous rounds, and secure key within previous QKD protocols were generated.
We formalize it by making $V_{i-2}$ and $V_{i-1}$ to be inputs for $\QKD{i}$ procedure for $i\geq 1$ and introducing the following rule:
If the owner of $V_j$ obtains $V_j=\bot$ then he or she sets the values of generated keys on his or here side within $\QKD{j+1}$ and $\QKD{j+2}$ to $\bot$.
We note that $V_j$ is generated on Bob's side for odd $j$, and on Alice's side for even $j$.
We also particularly note that the \QKD{1} procedure does not have any inputs.
It corresponds to the fact that the first QKD protocol always runs.

In order to perform a further authentication check the \QKD{i} outputs logs of resulting classical communication $\MA{i}$ and $\MB{i}$ used in the post-processing of corresponding QKD protocol.
Let $\TQ{i}\in \{0,1\}$ be a random variable corresponded to a tampering event within \QKD{i} procedure: $\TQ{i}=1$ in the case of a tampering (that is $\MA{i}\neq\MB{i}$), and $\TQ{i}=0$ -- otherwise.
Let us also introduce a random variable $\QE{i}\in\{0,1\}$ corresponded to an eavesdropping attempt within \QKD{i} procedure: $\QE{i}=0$ means that secure keys were generated successfully, and $\QE{i}=1$ -- otherwise.
We note that all technical imperfections which take place in the real implementation are considered to be produced by Eve.

We assume that the behavior of the $\QKD{i}^{\rm ideal}$ differs from the behavior of $\QKD{i}^{\rm real}$  only in the case of no tampering ($\TQ{i}=0$) and both parties participating in the QKD protocol ($V_{i-2}=V_{i-1}=\acc$ for $i\geq 2$).
The output $\rho^{\rm ideal}_{{\rm QKD}(i)}\leftarrow \QKD{i}^{\rm ideal}$ for $i \geq 2$ takes the form
	\begin{multline} \label{eq:QKDidealoutput}
	\rho^{\rm ideal}_{{\rm QKD}(i)} := \Pr[\QE{i}=0|\TQ{i}=0] \sigma^{(\Lsec)}_{\Akeyu{i}\Bkeyu{i}}\otimes \\ 	\sigma^{(\LOTP)}_{\KOTPAu{i+2}\KOTPBu{i+2}} \otimes \rho_{\MA{i}\MB{i}E}^{\QE{i}=0}+ \\
	\Pr[\QE{i}=1|\TQ{i}=0] \sigma^{\bot}_{\Akeyu{i}\Bkeyu{i}\KOTPAu{i+2}\KOTPBu{i+2}} \otimes 
	\rho_{\MA{i}\MB{i}E}^{\QE{i}=1},
\end{multline} 
where $\Lsec$ is a fixed length of a secure key for external applications, $\rho_{\MA{i}\MB{i}E}^{\QE{i}=0}$ and $\rho_{\MA{i}\MB{i}E}^{\QE{i}=1}$ are some states of registers \MA{i}, \MB{i} 
and Eve's system with only restriction that $\Pr[\MA{i}\neq\MB{i}]=0$ because of no tampering.
The output of $\QKD{1}^{\rm ideal}$ (in the case of no tampering) is the same as given in Eq.~\eqref{eq:QKDidealoutput}, but with the only difference that it also generates a $\Lrec$-bit recycled key.
It corresponds to an additional tensor factor $\sigma^{(\Lrec)}_{\KrecAu\KrecBu}$ in the case of $\QE{1}=0$ and $\sigma^{\bot}_{\KrecAu\KrecBu}$ in the case of $\QE{1}=1$ in Eq.~\eqref{eq:QKDidealoutput}.
The difference between the outputs of $\QKD{i}^{\rm ideal}$ and $\QKD{i}^{\rm real}$ is assumed to be limited by $\eps_{\rm QKD}$ in the Eq.~\eqref{eq:eps_QKD}.

\subsubsection{\AuthAB{i} and \AuthBA{i}} \label{apx:authcheck}
These procedure check the equality of \MA{i} and \MB{i} by generating an authentication tag with respect to \MA{i} transmitting it the opposite side, and verifying with respect to \MB{i}.
The superscript in the name of the procedure specifies a direction of the tag transfer, thus according to ping-pong pattern \AuthAB{i} is launched for odd $i$, \AuthBA{i} is launched for even $i$.
The main output of the authentication check procedure is the verification flag $V_i\in\{\bot,\acc\}$ which determines a continuation of the key growing process.

The input of the authenticity check procedure also includes the number of unverified quantum generated keys on the tag verifier side (see the Table~\ref{tab:protocols}).
If the verifier detects tampering  then he or she discards the unverified keys by setting the values of verified keys to $\bot$.
Otherwise, the verifying party sets the values of verified keys to be equal to corresponding values of unverified ones.
If as a result, fresh verified keys from \QKD{i} are not equal to $\bot$, then the verification flag $V_i$ is set to the value \acc, otherwise, it is set to $\bot$.
If the tag verifier of the $i^{\rm th}$ round obtains $V_i=\bot$ then he or she does not transmit an authentication tag in the $(i+1)^{\rm th}$ round.
If the tag verifier of the $i^{\rm th}$ round does not receive a message with a tag within a fixed timeout, then $V_i$ is set to $\bot$.

Let $\TA{i}\in\{0,1\}$ correspond to a tampering within an $i^{\rm th}$ authentication tag transfer: $\TA{i}=0$ in the case of no tampering, and $\TA{i}=1$ -- otherwise.
Let $\Tamp{i}$ be a flag of a tampering in $i^{\rm th}$ round: $\Tamp{i} := 0$ if $\TQ{i}=\TA{i}=0$, and $\Tamp{i} := 1$ -- otherwise. 
The ideal authenticity check detects a tampering and sets $V_i=\bot$ whenever $\Tamp{i}=1$.
In contrast, the real authenticity check based on $\eps_{\rm auth}$-ASU$_2$ family has a final probability of non-detection of an occurred tampering which does not exceed  $\eps_{\rm auth}$.
This fact is captured by Eq.~\eqref{eq:eps_auth}.
However, Ineq.~\eqref{eq:eps_auth} is valid assuming that the authentication keys used for tag generation and tag verification are ideal.
Remember, that the authentication tag is generated in the $i^{\rm th}$ round only in the case of $V_{i-1}=\acc$ (otherwise the tag sender transmits nothing), 
and is verified only in the case of $V_{i-2}=\acc$ (otherwise, the tag verifier sets $V_i:=\bot$ without any authenticity check).

Let $\Suc{i} := 1-\Tamp{i}\QE{i}$ be a flag of completing $i^{\rm th}$ round ($i\in\{1,\ldots,\Nmax\}$) without an Eve's attack either on classical or quantum channel.
To simplify a formal analysis we set $\Suc{\Nmax+1}:=0$ and $\Pr[\Tamp{\Nmax+1}=0]=1-\Pr[\Tamp{\Nmax+1}=1]=0$.
Let 
\begin{equation} \label{eq:Nsuc}
\Nsuc := \underset{i\in\{1,\ldots,\Nmax+1\}}{\rm argmin}\left\{ \Suc{i}=0 \right\}-1
\end{equation}
be a number of rounds before the first Eve's attack.

The following lemma enable us to rely on Ineq.~\eqref{eq:eps_auth} in the further security proof.
\begin{lemma} \label{lemma:1}
	Consider $\Nmax$-round key growing process whose procedures
	in the first $i-1$ rounds for $i\in\{2,\ldots,\Nmax\}$ are ideal.
	Let pre-distribution and storing procedures be ideal, while all other procedures be either real or ideal. 
	Let a tampering event occur at $i^{\rm th}$ round, that is $\Tamp{i}=1$.
	Then either
	\begin{itemize}
		\item $V_{i-1}=\acc$ and both the tag sender and the tag verifier of $i^{\rm th}$ round obtain identical secure authentication keys,
		\item or $V_{i-1}=\bot$ and the tag verifier of $i^{\rm th}$ round obtains secure authentication keys,
		\item or $V_{i-2}=\bot$ and the verifier of $i^{\rm th}$ round sets $V_i=\bot$ without any authenticity check.
	\end{itemize}
\end{lemma}
\begin{proof}
	First, consider the case $i=2$.
	Due to the fact, that pre-distribution and storing are ideal, the $\AuthBA{2}$ is supplied with ideal keys \KrecApredst and \KOTPApredst{2}.
	
	Then, to prove the Lemma for $i>2$, we consider different possible values of $\Nsuc$ (see Eq.~\eqref{eq:Nsuc}).
	Due to the fact that $\Tamp{i}=1$, we have $0\leq\Nsuc\leq i-1$.
	
	Suppose $\Nsuc=i-1$.
	It means that there was no Eve attack in any of the $i-1$ rounds, so both the tag sender and the tag verifier of the $i^{\rm th}$ round obtained ideal secure keys \KrecA, \KrecB from the first round and \KOTPA{i}, \KOTPA{i} from the $(i-2)^{\rm th}$ round.
	
	Suppose $\Nsuc=i-2$.
	It means that all $i-2$ rounds completed successfully, hence the verifier of the $i^{\rm th}$ round output $V_{i-2}=\acc$ and obtained uniformly random and secure authentication keys on his or her side.
	However, Eve's attack occurred in the $(i-1)^{\rm th}$ and, due to the fact that procedure within $(i-1)^{\rm th}$ round were ideal, the tag sender of $i^{\rm th}$ round output $V_{i-1}=\bot$.
	
	Finally, suppose $\Nsuc<i-2$.
	Due to the fact all procedures in the first $i-1$ rounds are ideal, the parties definitely obtain $V_{\Nsuc+1}=\ldots=V_{i-2}=V_{i-1}=\bot$.
\end{proof}

\subsection{Output of the ideal key growing process  }

Here we provide an explicit form of an output state generated in a key growing process consisted of ideal procedures.

\begin{lemma} \label{lemma:2}
	The $\Nmax$-round key growing process consisted of ideal procedures generates the state
	\begin{equation}\label{eq:totalstate}
	\rho^{\Nmax,{\rm ideal}}_{\rm total} = \sum_{i=0}^{\Nmax}\Pr[\Nsuc=i]\rho^{\Nsuc=i,{\rm ideal}}_{\rm total},
	\end{equation}
	where
	$\rho^{\Nsuc=0,{\rm ideal}}_{\rm total}:=
	\sigma^{\bot}[1]\otimes\ldots\otimes\sigma^{\bot}[\Nmax]\otimes
	\rho_E^{\Nsuc=0}$
	and
	\begin{multline}
		\rho^{\Nsuc=i,{\rm ideal}}_{\rm total} =\Pr[\Tamp{i+1}=0]\sigma^{\rm acc}[1]\otimes\ldots\otimes\sigma^{\rm acc}[i]\otimes \\
		\sigma^{\bot}[i+1]\otimes\ldots\otimes\sigma^{\bot}[\Nmax]\otimes
	\rho_E^{\Nsuc=i}+
	\\
	\Pr[\Tamp{i+1}=1]\sigma^{\rm acc}[1]\otimes\ldots\otimes\sigma^{\rm acc}[i-1]\otimes\sigma^{\sim}[i]\otimes \\
	\sigma^{\bot}[i+1]\otimes\ldots\otimes\sigma^{\bot}[\Nmax]\otimes\widetilde{\rho}_E^{\Nsuc=i}
	\end{multline}
	for $1\leq i\leq\Nmax$.
	Here ${\rho}_E^{\Nsuc=i}$ and $\widetilde{\rho}_E^{\Nsuc=i}$ are some quantum states of a system possessed by Eve. 
	Definitions of $\sigma^{\rm acc}[i]$ and $\sigma^{\bot}[i]$, $\sigma^{\sim}[i]$ are presented in the Table~\ref{tab:notations}.
\end{lemma}

\begin{table*} \centering
	\begin{tabular}{|c|c|c|c|} \hline
		$i$	&	
		$\sigma^{\rm acc}[i]$ & $\sigma^{\bot}[i]$ & 
		$\sigma^{\sim}[i]$ \\ \hline
		
		$1$	& 
		$\sigma^{(\Lrec)}_{\KrecA\KrecB}\otimes\sigma^{(\Lsec)}_{\Akey{1}\Bkey{1}}$ & 
		$\sigma^{\bot}_{\KrecA\KrecB\Akey{1}\Bkey{1}}$ & $\sigma^{(\Lsec)}_{\Bkey{1}}\otimes\sigma^{(\Lrec)}_{\KrecB}\otimes\sigma^{\bot}_{\KrecA\Akey{1}}$\\ \hline
		
		$2\ldots \Nmax-2$ & 
		$\sigma^{(\Lsec)}_{\Akey{i}\Bkey{i}}$ & $\sigma^{\bot}_{\Akey{i}\Bkey{i}}$ & 
		$\begin{aligned}
		\sigma^{(\Lsec)}_{\Akey{i}}\otimes\sigma^{\bot}_{\Bkey{i}} & \text{ for even }i\\
		\sigma^{\bot}_{\Akey{i}}\otimes\sigma^{(\Lsec)}_{\Bkey{i}} & \text{ for odd }i 
		\end{aligned}$ \\ \hline
		
		$\Nmax-1$ & 
		$\begin{aligned}
		&\sigma^{(\Lsec)}_{\Akey{\Nmax-1}\Bkey{\Nmax-1}}\otimes\\
		&\sigma^{(\LOTP)}_{\KOTPA{\Nmax+1}\KOTPB{\Nmax+1}}
		\end{aligned}$ &
		$\begin{aligned}
		&\sigma^{\bot}_{\Akey{\Nmax-1}\Bkey{\Nmax-1}}\otimes\\
		&\sigma^{\bot}_{\KOTPA{\Nmax+1}\KOTPB{\Nmax+1}}
		\end{aligned}$ &
		$\begin{aligned}
		& \sigma^{\bot}_{\Akey{\Nmax-1}\KOTPA{\Nmax+1}}\otimes \\
		& \sigma^{(\Lsec)}_{\Bkey{\Nmax-1}}\otimes\sigma^{(\LOTP)}_{\KOTPB{\Nmax+1}}
		\end{aligned}$ \\ \hline
		
		$\Nmax$ &
		$\begin{aligned}
		& \sigma^{(\Lsec)}_{\Akey{\Nmax}\Bkeyu{\Nmax}}\otimes \\
		& \sigma^{(\LOTP)}_{\KOTPA{\Nmax+2}\KOTPBu{\Nmax+2}}
		\end{aligned}$ &
		$\begin{aligned}
		&\sigma^{\bot}_{\Akey{\Nmax}\Bkeyu{\Nmax}}\otimes\\
		&\sigma^{\bot}_{\KOTPA{\Nmax+2}\KOTPBu{\Nmax+2}}
		\end{aligned}$ &
		\\\hline
	\end{tabular}
	\caption{Notation that are used for defining a state generated in the ideal key growing process.}
	\label{tab:notations}
\end{table*}
\begin{proof}
	To prove the Lemma we consider the output states generated with all possible values of $\Nsuc\in\{0,\ldots,\Nmax\}$.
	
	Suppose $\Nsuc=0$.
	It means that there was Eve's attack in the first round.
	Due to this fact, Bob obtains $V_1=\bot$ in \AuthAB{1} and sets all keys on his side within the whole key growing process to $\bot$.
	Then because of no valid tag in \AuthBA{2}, Alice discards generated keys from the first round and set all key has to be generated in the next rounds to $\bot$.
	As a result, Alice and Bob obtain the state $\sigma^{\bot}[1]\otimes\ldots\otimes\sigma^{\bot}[\Nmax-1]$, while Eve obtain some state, which we denote as $\rho_E^{\Nsuc=0}$.
	
	Suppose $1\leq\Nsuc\leq\Nmax{-1}$.
	It means that there was Eve's attack in $(\Nsuc+1)^{\rm th}$ round.
	Consider two cases: (i) an attack on quantum channel only, that is $\QE{\Nsuc+1}=1$, $\Tamp{i}=0$, and (ii) an attack on classical channel $\Tamp{i}=1$ 
	regardless of the value of $\QE{\Nsuc+1}$. 
	In both cases, the parties obtain ideal keys and in all rounds before the $\Nsuc^{\rm th}$ and obtain no key in all rounds after $\Nsuc^{\rm th}$. 
	It corresponds to the generated state $\sigma^{\acc}[1]\otimes\ldots\otimes\sigma^{\acc}[\Nsuc-1]\otimes \sigma^{\bot}[\Nsuc+1]\otimes\ldots\otimes\sigma^{\bot}[\Nmax]$.
	The difference appears for the keys generated in $\Nsuc^{\rm th}$ round.
	In the first case, the authentication tag in the authenticity check procedure of the $(\Nsuc+1)^{\rm th}$ round reaches the verifier, and though the key generation failed in $(\Nsuc+1)^{\rm th}$ round, 
	the verifier approves the keys generated in $\Nsuc^{\rm th}$ round.
	So, the state $\sigma^{\acc}[\Nsuc]$ generated.
	In the second case, the verifier of the $(\Nsuc+1)^{\rm th}$ round discovers a fact of attack but fails to understand in which particular round $\Nsuc^{\rm th}$ or $(\Nsuc+1)^{\rm th}$  the attack occurs.
	So the verifier discards keys from $\Nsuc^{\rm th}$ round as well and the state $\sigma^{\sim}[\Nsuc]$ is generated.
	We point out that in both cases Eve does not obtain any information about generated keys and obtains some arbitrary factorized state, which we denote with ${\rho}_E^{\Nsuc=i}$ or $\widetilde{\rho}_E^{\Nsuc=i}$.
	
	Finally, if $\Nsuc=\Nmax$, then there was no any attack during the whole key growing process and the parties obtain ideal keys.
	We note that Bob's keys from the final round remain to be unverified since the verifier of the final authenticity check procedure \AuthBA{\Nmax} is Alice.
	The state of Eve's system appears to be factorized with the state of Alice's and Bob's keys, so the resulting output state takes the form $\sigma^{\acc}[1]\otimes\ldots\otimes\sigma^{\acc}[\Nmax]\otimes{\rho}_E^{\Nsuc=\Nmax}$.
	
	Putting it all together and taking into account the fact that $\Nsuc$ is a random variable we obtain the state in the form of~Eq.~\eqref{eq:totalstate}. 	
\end{proof}

We note that in the output of the ideal key growing, the state of Eve's system is always factorized related to the shared state of Alice and Bob.
The only thing that Eve can do is interrupting the key growing.
Also note that a situation is possible when one of the legitimated parties generates a key being sure that it is secure, while the other one discards a corresponding key because of a tampering detection in the round (this situation is described by the state $\sigma^{\sim}[i]$).
This problem is common for every QKD protocol, regardless of whether each message is authenticated or delayed authentication is used:
One can think about a scenario where Eve blocks a last classical message within post-processing procedures.

\subsection{Proof of the Theorem~\ref{thm:mainsec}}\label{apx:mainsec}
\begin{proof}
	Consider a key growing process consisted of real procedures.
	The idea of the proof is to consistently replace real procedures with ideal ones and calculate distances between appearing intermediate states.
	Let $\rho^{\Nmax,{\rm mix}(0)}_{\rm total}$ be a state resulted from a new key growing process where real pre-distribution and storing procedures are replaced with ideal ones (see Fig.~\ref{fig:chain}).
	According to the properties of ideal procedures, one has
	\begin{equation}
		\Dist{\rho^{\Nmax,{\rm real}}_{\rm total}}{\rho^{\Nmax,{\rm mix}(0)}_{\rm total}}\leq \eps_{\rm pred}+\eps_{\rm store}.
	\end{equation}
	
	Then consider further replacing of $\QKD{1}^{\rm real}$ and $\AuthABr{1}$ with $\QKD{1}^{\rm ideal}$ and $\AuthABi{1}$.
	Let $\rho^{\Nmax,{\rm mix}(1)}_{\rm total}$ be the output state of this new process with four ideal procedures.
	
	According to the properties of replaced procedures, one has
	\begin{equation}
		\Dist{\rho^{\Nmax,{\rm mix}(0)}_{\rm total}}{\rho^{\Nmax,{\rm mix}(1)}_{\rm total}}\leq \eps_{\rm QKD}+\eps_{\rm auth}.
	\end{equation}
	We are able to replace \AuthABr{1} with \AuthABi{1} due to fact that it is supplied with ideal keys from \KeyPred\ideal.
	
	We then consider the following sequence of operations.
	Sequentially for each $i$ from 2 to \Nmax~we replace $\QKD{i}^{\rm real}$ and $\AuthABr{i}$ (or $\AuthBAr{i}$) with $\QKD{i}^{\rm ideal}$ and $\AuthABi{i}$ (or $\AuthBAi{i}$).
	Denote the resulted output state $\rho^{\Nmax,{\rm mix}(i)}_{\rm total}$.
	Next, according to the Lemma~\ref{lemma:2} one has
	\begin{equation}
		\Dist{\rho^{\Nmax,{\rm mix}(i)}_{\rm total}}{\rho^{\Nmax,{\rm mix}(i-1)}_{\rm total}}\leq \eps_{\rm QKD}+\eps_{\rm auth}.
	\end{equation}
	Here we also employ our assumption about distance between real and ideal procedures presented by Eq.~\eqref{eq:eps_QKD} and Eq.~\eqref{eq:eps_auth}
	
	Taking all together and noting that $\rho^{\Nmax,{\rm mix}(\Nmax)}_{\rm total}=\rho^{\Nmax,{\rm ideal}}_{\rm total}$ we obtain
	\begin{multline}
		\Dist{\rho^{\Nmax,{\rm ideal}}_{\rm total}}{\rho^{\Nmax,{\rm real}}_{\rm total}} \\ \leq \sum_{i=1}^{\Nmax}\Dist{\rho^{\Nmax,{\rm mix}(i)}_{\rm total}}{\rho^{\Nmax,{\rm mix}(i-1)}_{\rm total}}+\\\Dist{\rho^{\Nmax,{\rm mix}(0)}_{\rm total}}{\rho^{\Nmax,{\rm ideal}}_{\rm total}}\\=\eps_{\rm pred} + \eps_{\rm store} + \Nmax(\eps_{\rm auth}+\eps_{\rm QKD}).
	\end{multline}	
	This completes the proof.
\end{proof}


\begin{thebibliography}{99}

	\bibitem{BB84}
	C.H. Bennet and G. Brassard,
	{\href{http://researcher.watson.ibm.com/researcher/files/us-bennetc/BB84highest.pdf}
	{in {\it Proceedings of IEEE International Conference on Computers, Systems and Signal Processing} Bangalore, India (IEEE, New York, 1984), p. 175}}.
	
	\bibitem{Gisin2002}
	N. Gisin, G. Ribordy, W. Tittel, and H. Zbinden,
	{\href{https://dx.doi.org/10.1103/RevModPhys.74.145}{Rev. Mod. Phys. {\bf 74}, 145 (2002)}}.
	
	\bibitem{Scarani2009}
	V. Scarani, H. Bechmann-Pasquinucci, N. J. Cerf,  M. Dusek M, N. L\"utkenhaus, and M. Peev,
	\href{https://doi.org/10.1103/RevModPhys.81.1301}{Rev. Mod. Phys. {\bf 81}, 1301 (2009)}.
	
	\bibitem{Lo2015}
	H.-K. Lo, M. Curty, and K. Tamaki,
	{\href{https://dx.doi.org/10.1038/nphoton.2014.149}{Nat. Photonics {\bf 8}, 595 (2014)}}.
	
	\bibitem{Lo2016}
	E. Diamanti, H.-K. Lo, and Z. Yuan, 
	{\href{https://dx.doi.org/10.1038/npjqi.2016.25}{npj Quant. Inf. {\bf 2}, 16025 (2016)}}.
	
	\bibitem{Market}
	Characteristics of commercial devices from
	{\href{https://www.idquantique.com}{ID Quantique (Switzerland)}},
	{\href{https://www.sequrenet.com}{SeQureNet (France)}},
	and the {\href{https://www.ait.ac.at/}{Austrian Institute of Technology (Austria)}} are available.
	
	\bibitem{Larsson2016}
	C. Pacher, A. Abidin, T. Lor\"unser, M. Peev, R. Ursin, A. Zeilinger, and J.-{\AA}. Larsson,
	{\href{https://doi.org/10.1007/s11128-015-1160-4}{Quant. Inf. Proc. {\bf 15}, 327 (2016)}}.  
	
	\bibitem{Peev2004}
	M. Peev,  M. N\"olle, O. Maurhardt, T. Lor\"unser, M. Suda, A. Poppe, R. Ursin, A. Fedrizzi,  and A. Zeilinger,
	{\href{https://doi.org/10.1142/S0219749905000797}{Int. J. Quantum. Inf. {\bf 3}, 225 (2004)}}.
	
	\bibitem{Cederlof2008}
	J. Cederl\"of and J.-{\AA}. Larsson,
	{\href{https://doi.org/10.1109/TIT.2008.917697}{IEEE Trans. Inf. Theory {\bf 54}, 1735 (2008)}}.
		
	\bibitem{Peev2011}
	A. Abidin,  C. Pacher, T. Lor\"unser, J.-{\AA}. Larsson, and M. Peev,
	{\href{https://doi.org/10.1117/12.898344}{Proc. SPIE {\bf 8189}, 8189 (2011)}}.
	
	\bibitem{Jouguet2011}
	S. Kunz-Jacques and P. Jouguet,
	{\href{https://arxiv.org/abs/1109.2844}{arXiv:1109.2844}}.
	
	\bibitem{Abidin2011}
	A. Abidin and J-{\AA}. Larsson,
	{\href{https://doi.org/10.1007/978-3-642-34159-5_7}{Lect. Notes Comp. Sci. {\bf 7242}, 99 (2011)}}.
	
	\bibitem{Mosca2013}
	M. Mosca, D. Stebila, and B. Ustaoglu,
	{\href{https://doi.org/10.1007/978-3-642-38616-9_9}{Lect. Notes Comp. Sci. {\bf 7932}, 136  (2013)}}.
	
	\bibitem{Larsson2014}
	A. Abidin and J.-{\AA}. Larsson, 
	\href{10.1007/s11128-013-0641-6 }{Quant. Inf. Proc. {\bf 13}, 2155 (2014)}.
	
	\bibitem{Portmann2014}
	C. Portmann,
	{\href{http://dx.doi.org/10.1109/TIT.2014.2317312}{IEEE Trans. Inf. Theory {\bf 60}, 4383 (2014)}}.
	
	\bibitem{WegmanCarter1981}
	M.N. Wegman and J.L. Carter,
	{\href{http://dx.doi.org/10.1016/0022-0000(81)90033-7}{J. Comp. Syst. Sci. {\bf 22}, 265 (1981)}}.

	\bibitem{Stinson1994}
	D.R. Stinson,
	{\href{http://dx.doi.org/10.1007/3-540-46766-1_5}{Des., Codes Cryptography {\bf 4}, 369 (1994)}}.	

	\bibitem{Vernam1926}
	G.S. Vernam,
	{\href{http://www.doc.ic.ac.uk/~mrh/330tutor/vernam.pdf}{J. Amer. Inst. Electr. Engineers {\bf 45}, 109 (1926)}}.
	
	\bibitem{Quade2009}
	J. M\"uller-Quade and R. Renner,
	\href{https://doi.org/10.1088/1367-2630/11/8/085006}{New J. Phys. {\bf 11}, 085006 (2009)}.
	
	\bibitem{Maurhart2010}
	O. Maurhart, 
	\href{https://doi.org/10.1007/978-3-642-04831-9_8}
	{Lect. Notes Phys. {\bf 797}, 151 (2010)}.	

	\bibitem{Peev2009}
	M. Peev et al., 
	\href{http://dx.doi.org/10.1088/1367-2630/11/7/075001}{New J. Phys. {\bf 11}, 075001 (2009)}.
	
	\bibitem{BenOr2005}
	M. Ben-Or, M. Horodecki, D. W. Leung, D. Mayers, and J. Oppenheim,
	\href{https://doi.org/10.1007/978-3-540-30576-7_21}{Lect. Notes Comp. Sci. {\bf 3378}, 386  (2005)}.
	
	\bibitem{Schneier1996}
	B. Schneier,
	{\it Applied cryptography} 
	(John Wiley \& Sons, Inc., New York, 1996).
	
	\bibitem{Stinson2002}
	D.R. Stinson, 
	\href{}{J. Combin. Math. Combin. Comput. {\bf 42}, 3 (2002)}.
	
	\bibitem{Boer1993}
	B.D. Boer,
	\href{}{J. Comp. Sec. {\bf 2}, 65 (1993)}.
	
	\bibitem{Krawczyk1994}
	H. Krawczyk, 
	\href{https://doi.org/10.1007/3-540-48329-2_28}{Lect. Notes Comput. Sci. {\bf 839}, 129 (1994)}.
	
	\bibitem{Lemire2014}
	D. Lemire and O. Kaser,
	\href{https://doi.org/10.1093/comjnl/bxt070}{Comput. J. {\bf 57}, 1624 (2014)}. 
	
	\bibitem{Bierbrauer1994}
	J. Bierbrauer, T. Johansson, G. Kabatianskii, and B. Smeets, 
	\href{https://doi.org/10.1007/3-540-48329-2_28}{Lect. Notes Comput. Sci. {\bf 773}, 331 (1994)}.	
	
	\bibitem{WegmanCarter1979}
	M.N. Wegman and J.L. Carter,
	{\href{https://doi.org/10.1016/0022-0000(79)90044-8}{J. Comp. Syst. Sci. {\bf 18}, 143 (1979)}}.
	
	\bibitem{Kiktenko2016}
	E. O. Kiktenko, A. S. Trushechkin, Y. V. Kurochkin, A. K. Fedorov
	\href{https://doi.org/10.1088/1742-6596/741/1/012081}{J. Phys. Conf. Ser. {\bf 741}, 012081 (2016)}.
	
	\bibitem{Stinson1996}
	D.R. Stinson,
	\href{http://citeseerx.ist.psu.edu/viewdoc/summary?doi=10.1.1.50.5465}{Congressus Numerantium {\bf 114}, 7 (1996)}.
	
	\bibitem{Krovetz2000}
	T. Krovetz and P. Rogaway,
	\href{https://doi.org/10.1007/3-540-45247-8_7}{Lect. Notes Comput. Sci. {\bf 2015}, 73 (2000)}.
	
	\bibitem{Walenta2014}
	N. Walenta, A. Burg, D. Caselunghe, J. Constantin, N. Gisin, O. Guinnard, R. Houlmann, P. Junod, B. Korzh, N. Kulesza, M. Legr\'e, 
	C.C.W. Lim, T. Lunghi, L. Monat, C. Portmann, M. Soucarros, P. Trinkler, G. Trolliet, F. Vannel, and H. Zbinden,
	{\href{https://dx.doi.org/10.1088/1367-2630/16/1/013047}{New J. Phys. {\bf 16} 013047 (2014)}}.
	
	\bibitem{Tomamichel2017}
	M. Tomamichel and A. Leverrier,
	\href{https://doi.org/10.22331/q-2017-07-14-14}{Quantum {\bf 1}, 14 (2017)}.

	\bibitem{DH1976}
	W. Diffie and M. Hellman,
	\href{https://doi.org/10.1109/TIT.1976.1055638}{IEEE T. Inform. Theory {\bf  22}, 644}.

	\bibitem{Shor1997}
	P.W. Shor, 
	\href{https://doi.org/10.1137/S0036144598347011}{SIAM J. Comput. {\bf 26}, 1484 (1997)}.

	\bibitem{Code}
	E.O. Kiktenko, A.O. Malyshev, M.A. Gavreev, A.A. Bozhedarov, N.O. Pozhar, M.N. Anufriev, and A.K. Fedorov,
	{\href{https://dx.doi.org/10.5281/zenodo.2582712}{Lightweight authentication for quantum key distribution (source code)}},
	DOI: 10.5281/zenodo.2582712.

	
\end{thebibliography}
\end{document}